\documentclass[11pt, onecolumn, letter]{IEEEtran}
\usepackage{latexsym}
\usepackage{amssymb}
\usepackage{amsmath}
\usepackage{amsthm}
\usepackage[dvips]{graphicx}
\usepackage{multirow}

\usepackage{color}

\newcommand{\san}[1]{\mathsf{#1}}
\newcommand{\rom}[1]{\mathrm{#1}}

\newcommand{\Pj}{\rom{P}_{\rom{j}}}
\newcommand{\Pjs}{\rom{P}_{\rom{js}}}
\newcommand{\Ps}{\rom{P}_{\rom{s}}}
\newcommand{\Pc}{\rom{P}_{\rom{c}}}
\newcommand{\K}{\rom{K}}
\newcommand{\V}{\rom{V}}

\newtheorem{definition}{Definition}
\newtheorem{proposition}[definition]{Proposition}
\newtheorem{theorem}{Theorem}
\newtheorem{corollary}{Corollary}
\newtheorem{remark}[definition]{Remark}
\newtheorem{lemma}{Lemma}

\def\Label#1{\label{#1}\ [\ \text{#1}\ ]\ }
\def\Label{\label}

\begin{document}

\title{Second Order Analysis for Joint Source-Channel Coding with Markovian Source}
\author{Ryo Yaguchi and Masahito Hayashi~\IEEEmembership{Fellow,~IEEE}
\thanks{The material in this paper will be presented in part at the 2017 IEEE International Symposium on Information Theory (ISIT 2017),   Aachen (Germany), 25-30 June 2017.}
\thanks{Ryo Yaguchi  was with the Graduate School of Mathematics, Nagoya University,
Furocho, Chikusaku, Nagoya, 464-860, Japan}
\thanks{Masahito Hayashi   is with the Graduate School of Mathematics, Nagoya University,
Furocho, Chikusaku, Nagoya, 464-860, Japan,
and
Centre for Quantum Technologies, National University of Singapore, 3 Science Drive 2, Singapore 117542.
(e-mail: masahito@math.nagoya-u.ac.jp)}
}

\markboth{R. Yaguchi and M. Hayashi: Second Order Analysis for Joint Source-Channel Coding}{}

\maketitle

\begin{abstract}
We derive the second order rates of joint source-channel coding, whose source obeys an irreducible and ergodic Markov process when the channel is a discrete memoryless,
while a previous study solved it only in a special case.
We also compare the joint source-channel scheme with the separation scheme in the second 
order regime while a previous study made a notable comparison only with numerical calculation.
To make these two notable progress, 
we introduce two kinds of new distribution families, 
switched Gaussian convolution distribution
and $*$-product distribution, which are defined by modifying the Gaussian distribution. 
\end{abstract}


\section{Introduction}\Label{S1}
Nowadays, second order analysis attracts much attention in information theory \cite{7,6,HSO,12,DAY}.
In this type of analysis, 
we focus on the second leading term with the order $\sqrt{n}$ in the coding length 
in addition to the first leading term with the order $n$ when the block length is $n$.
To discuss the finiteness of the blocklength, we need to be careful for 
the second leading term as well as the first leading term.
The coefficient of the order $\sqrt{n}$ is given 
as the inverse of the cumulative distribution function of the Gaussian distribution
depending on the decoding error probability $\varepsilon$
in many existing studies for the second order except for the papers \cite{KH1,KH2}.
This is because the second order analysis is deeply rooted in the central limit theorem.
In channel coding, the second order coefficient is given by the Gaussian distribution,
whose variance is given as the variance of the information density.
Here, the information density is given as the logarithm of the likelihood ratio between 
the joint distribution of the input and output random variable 
and their product distribution 
when the expectation of the logarithm of the likelihood ratio achieve the channel capacity.
However, the variance of the information density is not unique, in general
because multiple input distributions attain the channel capacity in general.
So, in such a general case, the variance of the Gaussian determining the second order coefficient is chosen depending on the sign of the decoding error probability $\varepsilon$.
Recently, the two papers \cite{HW,HW2} extended the second order analysis to the Makovian case,
in which, the Markovian version of the central limit theorem is employed instead of the conventional central limit theorem.
In particular, the paper \cite{HW} discussed 
source coding for Markovian source and channel coding for additive channel whose additive noise is Markovian.
Also, Kontoyiannis and Verd\'{u}, \cite{KV14} discussed the variable-length source coding in a similar setting.

Usually, the channel coding is discussed with the message subject to the uniform distribution.
However, in the real communication, the message is not necessarily subject to the uniform distribution.
To resolve this problem, we often consider the channel coding with 
the message subject to the non-uniform distribution.
Such a problem is called source-channel joint coding and has been actively studied by several researchers \cite{Csis,CVFKM,ZA,DAY,KV,VSM}.
As a simple case, we often assume that the message is subject to the independent and identical distribution.
In this case, the capacity is given as the ratio of
the conventional channel capacity to the entropy of the message.
Several studies \cite{Csis,CVFKM,ZA} derived the exponential decreasing rate of the decoding error probability in this setting.
Recently, while Wang-Ingber-Kochman \cite{DAY} and Kostina-Verd\'{u} \cite{KV,KV13} discussed the second-order coefficient in this problem,
two major open problems has been remained in this topic as follows.
Wang-Ingber-Kochman \cite{DAY} derived the second order coefficient
only when the variance of the information density is unique.
When the variance is not unique, 
Kostina-Verd\'{u} \cite{KV13} extended it to the lossy case.
Kostina-Verd\'{u} \cite{KV} extended the lower bound of the second-order coefficient 
by the same method as \cite{DAY}.
However, the impossibility to improve the bound has been an open problem in the general case.
Also, in the above special case,
Wang-Ingber-Kochman \cite{DAY} compared their second order coefficient of the joint scheme
with that with the separation scheme.
Based on their numerical calculation,
they conjectured an inequality for the loss of the separation scheme \cite{DAY2}, whose analytical proof has been remained as another open problem.

In this paper, we tackle both open problems.
Firstly, we derive the second-order coefficient in this problem.
The obtained coefficient is strictly larger than that by Kostina-Verd\'{u} \cite{KV}
when the variance of the information density is not unique.
To characterize the second-order coefficient,
we introduce a new probability distribution as a generalization of the Gaussian distribution.
That is, the second-order coefficient is given
as the inverse of the cumulative distribution function of the new probability distribution.
Further, we derive this result even when the distribution of the message is Markovian.
Secondly, we discuss the second order coefficient with the separation scheme
in the above general setting.
Also, we analytically determine the range of the ratio between 
the error probabilities with the joint and separation schemes
when the variance of the information density is unique.
In this way, we resolve both open problems.

The remaining part of this paper is organized as follows.
In Section \ref{S2}, we prepare several information quantities for Markovian process.
Section \ref{S3} introduces two new distribution families.
In Section \ref{S4}, we discuss the  joint source-channel coding in the single shot setting.
Then, 
Section \ref{S5} 
shows our results for Markovian conditional additive channel.
discusses the second order rate.
Section \ref{S6} 
discusses the case with discrete memoryless channel.
In Section \ref{S7}, we compare the joint source-channel scheme with the separation scheme.

\section{Notations and Information quantities}\Label{S2}
\subsection{Single shot}\Label{S21}
In this paper, 
we denote the random variable by a capital letter, e.g., $X$.
By ${\cal X}$, we denote the set that the random variable $X$ takes values in.
Then, we denote the distribution of the random variable $X$ by $P_X$.
When we have two distributions $P_X$ and $P_Y$, 
we define their product distribution $P_X \times P_Y$
as $(P_X \times P_Y)(x,y):=P_X(x)P_Y(y)$.

When we have two different sets ${\cal X}$ and ${\cal Y}$,
we denote a transition matrix from ${\cal X}$ to ${\cal Y}$ by $W_{Y|X}$.
Then, we define the distribution $P_X \times W_{Y|X}$
as $(P_X \times W_{Y|X})(x,y)=P_X(x)W_{Y|X}(y|x)$.
When ${\cal X}$ is the same set as ${\cal Y}$,
we do not describe the subscript $Y|X$.
In this case,
we define the transition matrix $W^n$ on ${\cal X}$
as $W^n(x_{n}|x_0):= \sum_{x_{n-1},\ldots x_1}
W(x_n|x_{n-1})W(x_{n-1}|x_{n-2})\cdots W(x_1|x_{0})$.
A transition matrix $W$ on ${\cal X}$ is called {\it irreducible} 
when for each $x,x'\in {\cal X}$, there exists a natural number $n$ such that $W^n(x|x')>0$.
An irreducible matrix $W$ is called {\it ergodic} when 
there are no input $x'$ and no integer $n'$
such that $W^{n} (x'|x')=0$ unless 
$n$ is divisible by $n'$.

\subsection{Markovian process}\Label{S22}
Since this paper addresses the Markovian processes, 
we prepare several information measures given in \cite{HW} for an ergodic and irreducible transition matrix $W = \{W(x, z|x', z') \}_{(x, z), (x', z') \in ({\cal X \times Z} )^2 }$ on $({\cal X} \times {\cal Z})$. 
For this purpose, we employ the following assumption on transition matrices, which were introduced by the paper \cite {HW}. 

\begin{definition}[non-hidden]
When an ergodic and irreducible transition matrix $W $ satisfies the condition
\begin{align}
\sum_{x} W(x, z|x', z') = W(z|z')
\end{align}
for every $x' \in {\cal X}$ and $z, z' \in {\cal Z}$,
it is canned non-hidden (with respect to $ {\cal Z} $).
\end{definition}

For example, when the cardinality of ${\cal Z}$ is $1$, the above non-hidden condition holds.
For a non-hidden transition matrix $W$ on ${\cal X} \times {\cal Z}$ with respect to $ {\cal Z} $,
we define the marginal $W_Z$ by $W_Z(z|z') := \sum_x W(x, z|x',z')$. 
In the following, we assume the non-hidden condition.
By $\lambda_{\theta}$, 
we denote the Perron-Frobenius eigenvalue of
\begin{align}
W(x, z|x', z')^{1+\theta}W_Z(z|z')^{-\theta}
\end{align}
for a real number $\theta$.
Then, we define the conditional R\'{e}nyi entropy for the transition matrix \cite{HW} as
\begin{align}
H_{1+\theta}^{W, \downarrow}(X|Z) :=
-\frac{1}{\theta} \log \lambda_{\theta}, 
\end{align}
which is often called 
the lower type of conditional R\'{e}nyi entropy and is denoted by 
$H_{1+\theta}^{W, \downarrow}(X|Z)$ in \cite{HW}.

Taking the limit $\theta \to 0$,
we define the entropy for the transition matrix $W $ as
\begin{eqnarray}
H^W(X|Z):=
\lim_{ \theta \to 0 }H_{ 1 + \theta }^W(X|Z).  \Label{Hw}
\end{eqnarray}
To discuss the difference of $H_{ 1 + \theta }^W(X|Z) $ from $H^W(X|Z)$,
we introduce the varentropy for the transition matrix
$ \Gamma $ as
\begin{eqnarray}
V^W(X|Z )
:=
\lim_{ \theta \to 0 }
\frac{2[H^W(X|Z) - H_{ 1 + \theta }^W(X|Z )]} { \theta }. \Label{Vw}
\end{eqnarray}
So, we have the approximation as
$H_{ 1 + \theta }^W(Z|X )= H^W(Z|X )- \frac{1}{2} V^W(Z|X  )\theta +O(\theta^2)$ as $\theta \to 0$.
In these definitions,
when the output distribution of $W$ does not depend on the input element,
 the quantities
$H_{ 1 + \theta }^W(X|Z)$, $H^W(X|Z)$, and $V^W(X|Z)$
are the same as the conventional definitions.
Then, we have the following proposition.

\begin{proposition}[Central limit theorem for Markovian Process (\cite{A}etc.)]
\Label{CLT}
When $X^n=(X_1, \ldots, X_n)$ and $Z^n=(Z_1, \ldots, Z_n)$
are subject to the Markovian process generated by a non-hidden transition matrix $W$,
the random variable 
$\frac{1}{\sqrt{n}}
(- \log P_{X^n|Z^n}(X^n|Z^n)-n  H^W(X|Z))$
asymptotically obeys the Gaussian distribution with variance $V^W(X|Z)$\footnote{There are so many 
literatures for central limit theorem for Markovian Process.
The paper \cite[Corollary 7.2.]{WH} gives its very elementary proof.
It also summarizes existing approaches for this statement.}.
\end{proposition}

\section{New Probability Distribution Families}\Label{S3}
\subsection{Switched Gaussian convolution distribution}\Label{S31}
To describe the second order rate in the joint source-channel coding,
we introduce a new type of distribution family, so called switched Gaussian convolution distributions.
It is known that the convolution of two Gaussian distributions is also a Gaussian distribution as follows.
When $\varphi_v$ is the probability density function of the Gaussian distribution with average $0$ and variance $v$,
we have
\begin{align}
\varphi_{v_1+v_2}(x)
=\int_{-\infty}^{\infty}\varphi_{v_1}(y) \varphi_{v_2}(x-y) dy.
\end{align}
Now, we consider the case when 
the variance of the second probability density function is switched at $y=x$.
So, we define the function $\psi[v_1,v_2,v_3](x)$ as
\begin{align}
&\psi[v_1,v_2,v_3](x)\nonumber \\
:=&\int_{-\infty}^{x}\varphi_{v_1}(y) \varphi_{v_+}(x-y) dy
+\int_{x}^{\infty}\varphi_{v_1}(y) \varphi_{v_-}(x-y) dy,
\end{align}
where $v_+:=\max\{ v_2,v_3\}$ and $v_-:=\min\{ v_2,v_3\}$.
Taking the integral with respect to $x$, we define the function 
${\Psi}[v_1,v_2,v_3](R):=\int_{-\infty}^R \psi[v_1,v_2,v_3](x) dx$, which satisfies
\begin{align}
&{\Psi}[v_1,v_2,v_3](R)\nonumber \\
=&\int_{-\infty}^{R}\varphi_{v_1}(y) \Phi_{v_+}(R-y) dy
+\int_{R}^{\infty}\varphi_{v_1}(y) \Phi_{v_-}(R-y) dy \nonumber \\
=&\int_{-\infty}^{\infty}\varphi_{v_1}(y) 
\min\{\Phi_{v_2}(R-y), \Phi_{v_3}(R-y) \} dy ,\Label{12-24}
\end{align}
where $\Phi_{v}(R):=\int_{-\infty}^R\varphi_v(x) dx$.
We simplify $\Phi_{v}$ to $\Phi$ when $v=1$.

Since the value $\min\{\Phi_{v_2}(R-y), \Phi_{v_3}(R-y) \} $ goes to $0$($1$) as $R$ goes to $-\infty$($\infty$),
respectively,
the RHS of  \eqref{12-24} goes to 
$0$($1$) as $R$ goes to $-\infty$($\infty$), respectively,
Also, the value $\min\{\Phi_{v_2}(R-y), \Phi_{v_3}(R-y) \} $ is monotonically increasing with respect to $R$,
the RHS of  \eqref{12-24} also is monotonically increasing with respect to $R$.
These facts show that 
${\Psi}[v_1,v_2,v_3](R)$ is the cumulative distribution function of a probability distribution.
In the following, we call this distribution 
the switched Gaussian convolution distribution with $v_1,v_2$, and $v_3$.
 
To see the behavior of the distribution function of 
the switched Gaussian convolution distribution,
we set $v_1=v_2=1$, and change the third parameter $v_3$.
Then, we obtain the graph given in Fig. \ref{F3}.
From the definition, we find that
the maximum $\max_{v_3}{\Psi}[1,1,v_3](x)$ is realized when $v_3=1$.
Fig. \ref{F3} shows how much ${\Psi}[1,1,v_3](x)$ decreases unless $v_3=1$.

\begin{figure}[htbp]
\begin{center}
\scalebox{0.7}{\includegraphics[scale=1.3]{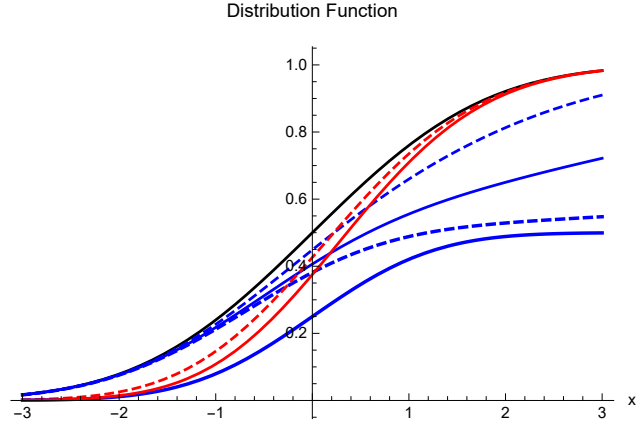}}
\end{center}
\caption{Graphs of ${\Psi}[1,1,v_3](x)$.
Black line: $v_3=1$,
Red dashed line: $v_3=1/9$,
Red normal line: $v_3 \to 0$,
Blue dashed line: $v_3=4$,
Blue normal line: $v_3=25$,
Blue dashed thick line: $v_3=25^2$,
Blue thick line: $v_3 \to \infty$.
}
\Label{F3}
\end{figure}%

\subsection{$*$-product distribution}\Label{S32}
Now, given two parameter $v_1,v_2>0$, we define another probability distribution.
For this purpose, we define the function $\tilde{\Phi}[v_1,v_2]$ as 
\begin{align}
\tilde{\Phi}[v_1,v_2](R)
:=\min_{a\in \mathbb{R}}
\Phi_{v_1}(a)
*  \Phi_{v_2}(R-a)  ,\Label{sep,errX}
\end{align}
where the product $*$ is defined as 
\[a*b=a+b-a b.\]
So, the inverse function
$\tilde{\Phi}[v_1,v_2]^{-1}$ 
is given as
\begin{align}
\tilde{\Phi}[v_1,v_2]^{-1}(\varepsilon)
=
\max_{\varepsilon = \varepsilon_s * \varepsilon_c  }
\left (
\Phi_{v_1}^{-1}(\varepsilon_s)
 +  \Phi_{v_2}^{-1}(\varepsilon_c) 
\right ) .\Label{sep,err}
\end{align}
Since the function $\tilde{\Phi}[v_1,v_2]$  satisfies the condition of the cumulative distribution function,
it can be regarded as the cumulative distribution function of another probability distribution.
We call it $*$-product distribution because it is defined based on the $*$ product.

The cumulative distribution function $\tilde{\Phi}[v_1,v_2]$ has the following property. 
\begin{lemma}\Label{sep,so,lem}
For any 
$ v_1,v_2>0  $, we have 
\begin{align}
\Phi_{v_1+v_2}(R)
\le 
\tilde{\Phi}[v_1,v_2](R)
\le
2 \Phi_{2(v_1+v_2)}(R)
-\Phi_{2(v_1+v_2)}(R)^2.
\Label{sep,so,th1}
\end{align}
The equality in the first inequality
is attained 
if and only if $ \frac{v_1}{v_2} $ is $0$ or $\infty$.
When $R\le 0$, the equality of the second inequality is attained if and only if $ v_1=v_2 $.
\end{lemma}
Lemma \ref{sep,so,lem} is shown in Appendix \ref{Ap3}.
The functions in Lemma \ref{sep,so,lem} are numerically compared in Fig. \ref{v2=01}.
When $ v_1=v_2 $,
we also numerically checked that the equality of the second inequality holds 
even for $R>0$.
Overall,  the cumulative distribution functions of this paper are summarized in Table \ref{T1}.

\begin{figure}[htbp]
\begin{center}
\scalebox{0.7}{\includegraphics[scale=1.3]{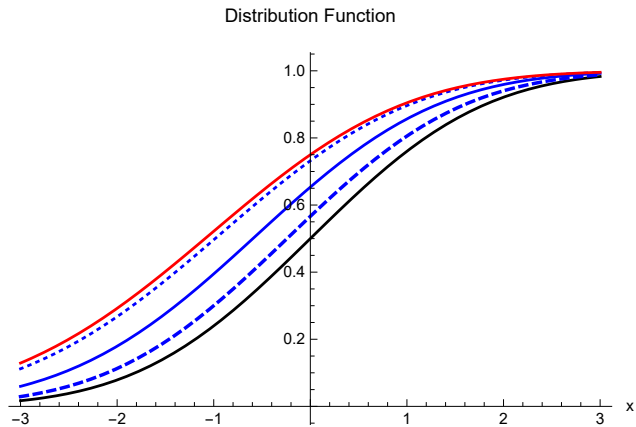}}
\end{center}
\caption{Graphs of functions in \eqref{sep,so,th1}.
Red line: $2 \Phi_{4}(x)-\Phi_{4}(x)^2$.
Blue dotted line: $\tilde{\Phi}[1.5,0.5](x)$.
Blue normal line: $\tilde{\Phi}[1.9,0.1](x)$.
Blue dashed line: $\tilde{\Phi}[1.99,0.01](x)$.
Black line: $\Phi_{2}(x)$.
}
\Label{v2=01}
\end{figure}%

\begin{table}[htpb]
  \caption{Cumulative distribution functions}
\Label{T1}
\begin{center}
  \begin{tabular}{|c|l|} 
\hline
$\Phi_v $ & Gaussian distribution with variance $v$ \\
\hline
$\Psi[v_1,v_2,v_3] $ & Switched Gaussian convolution distribution \eqref{12-24} \\
\hline
$\tilde{\Phi}[v_1,v_2]$ & $*$-product distribution \eqref{sep,err}\\
\hline
  \end{tabular}
\end{center}
\end{table}

\begin{remark}
The paper \cite[Section V]{DAY} considered  the function $\tilde{\Phi}[1,v_2]^{-1}(\varepsilon)/\sqrt{1+v_2}$, and gave the same statement as the second inequality in \eqref{sep,so,th1} with 
the condition $\tilde{\Phi}[v_1,v_2](R) < \frac{3}{4}$
in a difference form as a conjecture based on numerical calculations.
This conjecture had been an open problem.
\end{remark}

\section{Single Shot Setting}\Label{S4}

\subsection{Problem formulation}\Label{S41}
We first present the problem formulation by the single shot setting. 
Assume that the message $M$ takes values in ${\cal M}$
and is subject to the distribution $P_M$. 
For a channel $W_{Y|X}(y|x)$ with input alphabet ${\cal X}$
and output alphabet ${\cal Y}$, 
a channel code $\phi = (\san{e}, \san{d})$ consists of one encoder 
$\san{e}: {\cal M} \to {\cal X}$ and
one decoder $\san{d}:{\cal Y} \to {\cal M}$. 
The average decoding error probability is defined by
\begin{eqnarray}
\Pjs[\phi|P_M, W_{Y|X}] := \sum_{m \in {\cal M}}
P_M(m) W_{Y|X}(\{b:\san{d}(b) \neq m \}|\san{e}(m)). 
\end{eqnarray}
For notational convenience, we introduce the smallest attainable decoding error probability 
under the above condition:
\begin{eqnarray}
\Pjs(P_M, W_{Y|X}) := \inf_{\phi} \Pjs[\phi|P_M, W_{Y|X}]. 
\end{eqnarray}

\subsection{Direct part}\Label{S42}
\subsubsection{General case}
We introduce several lemmas for the case when
${\cal M}$ is the set of messages to be sent, 
$P_M$ is the distribution of the messages, and
$W_{Y|X}$ is the channel from ${\cal X}$ to ${\cal Y}$. 
We have the following single-shot lemma for the direct part. 

\begin{proposition}\cite[Lemma 3.8.1]{han} \Label {si,di,le}
For any constant $c>0$ and for any $P_X \in {\cal P(X)} $, 
there exists a code $ \phi = (\san{e}, \san{d}) $ such that
\begin{align}
\Pjs[\phi|P_M, W_{Y|X}] \le
(P_M \times P_X \times W_{Y|X})
\{ (P_M \times P_X \times W_{Y|X}) (M, X, Y)
\le c (P_X \times \bar{W}_{Y} )(X, Y)
\}+ \frac{1}{c}, \Label{si,di,le1}
\end{align}
where $\bar{W}_{Y}(y):= \sum_{x}P_X(x) W_{Y|X}(y|x)$
and $P_X \times W_{Y|X} (y, x):= P_X (x) W_{Y|X} (y|x)$. 
\end{proposition}

From above Proposition, we obviously have the following corollary.
\begin{corollary}\Label{Co11}
\begin{align}
\Pjs(P_M, W_{Y|X}) \le
(P_M \times P_X \times W_{Y|X})
\{ (P_M \times P_X \times W_{Y|X}) (M, X, Y)
\le c (P_X \times \bar{W}_{Y} )(X, Y)
\}+ \frac{1}{c}. \Label{3}
\end{align}
\end{corollary}

\subsubsection{Conditional additive case}
Now, we proceed to the case when the channel is conditional additive. 
Assume that ${\cal X}$ is a module and 
${\cal Y}$ is given as ${\cal X}\times {\cal Z}$. 
Here, $Z$ is called the internal state.
Then, the channel $W$ is called conditional additive \cite{HW} when
there exists a joint distribution $P_{XZ}$ such that
\begin{align}
W_{XZ|X}(x, z|x')= P_{XZ}(x-x', z). \Label{ca1}
\end{align}

We summarize the relation between general case and conditional additive case as Table \ref{T2}.

\begin{table}[htpb]
  \caption{Relation between general case and conditional additive case}
\Label{T2}
\begin{center}
  \begin{tabular}{|l|c|c|} 
\hline
& general case & conditional additive\\
\hline
message  & $M$ & $M$  \\
\hline
input & $X$ & $X$  \\
\hline
output variable & $Y$ & $(X,Z)$ \\
\hline
channel & $W_{Y|X}$ & $W_{XZ|X}$ \\
\hline
encoder  & $\san{e}$ & $\san{e}$ \\
\hline
decoder & $\san{d}$ & $\san{d}$ \\
\hline
distribution  & \multirow{2}{*}{$P_{M}$} & \multirow{2}{*}{$P_{M}$} \\
of message & & \\\hline
decoding error & \multirow{2}{*}{$\Pjs(\phi|P_M, W_{Y|X})$} & 
\multirow{2}{*}{$\Pjs(\phi|P_M,W_{XZ|X})$}  \\
probability & & \\
\hline
  \end{tabular}
\end{center}
\end{table}

Then we simplify (\ref{3}) of Corollary \ref{Co11} to the following lemma. 
\begin{lemma}
A conditional additive channel $W_{XZ|X}$ satisfies the inequality 
\begin{align}
\Pjs(P_M, W_{XZ|X}) \le
P_M \times P_{XZ} 
\{ P_M(M)P_{X|Z}(X|Z)
\le c \frac{1}{|{\cal X}|}
\} + \frac{1}{c}. \Label{a}
\end{align}
\end{lemma}
\begin{proof}
By setting that $P_X$ is the uniform distribution 
and choosing the random variables $X=X'$ and $Y=XZ$ 
to the right hand side of (\ref{3}), 
we have 
\begin{align*}
& (P_M \times P_{X'} \times W_{XZ|X'})
\{ 
(P_M \times P_{X'} \times W_{XZ|X}) (M, X', XZ)
\le c 
P_{X'} \times \bar{W}_{XZ} (X', X, Z)\} \\
= &
(P_M \times P_{X'} \times W_{XZ|X})
\{
P_M(m) \frac{1}{|{\cal X}|} 
P_{XZ}(x-x', z)
\le c 
\frac{1}{|{\cal X}|^2}
P_{Z}(z)
\} \\
= &
(P_M \times P_X \times W_{XZ|X'})
\{
P_M(m)P_{X|Z}(x-x'|z)
\le c 
\frac{1}{|{\cal X}|}
\} \\
= &
P_M \times P_{XZ} 
\{ P_M(M)P_{X|Z}(X|Z)
\le c \frac{1}{|{\cal X}|}
\}, 
\end{align*}
where $ P_Z (z):= \sum_{x} P_{XZ} (x, z) $. 
Hence, 
(\ref{3}) can be simplified to 
\begin{align}
\Pjs(\phi|P_M, W_{Y|X}) \le
P_M \times P_{XZ} 
\{ P_M(M)P_{X|Z}(X|Z)
\le c \frac{1}{|{\cal X}|}
\} + \frac{1}{c}.
\end{align}
\end{proof}

\subsection{Converse part}\Label{S43}
\subsubsection{General case}
Firstly, combining the idea of meta converse \cite{Naga}\cite[Lemma 4]{HN}\cite{12} 
and the general converse lemma for the joint source and channel coding \cite[Lemma 3.8.2]{han}, we obtain the following lemma for the single shot setting. 
The following lemma is the same as \cite[Lemma 3.8.2]{han} when $Q_Y$ is $\bar{W}_Y$. 

\begin{lemma}\Label{L-11}
For any constant $c>0$, any code $ \phi = (\san{e}, \san{d}) $ and any distribution $ Q_Y $ on $ {\cal Y} $, we have
\begin{align}
\Pjs(P_M, W_{Y|X}) \ge
\sum_m P_M(m)W_{Y|X=\san{e}(m)} 
\{ P_M(m)W_{Y|X=\san{e}(m)} (Y)
\le c Q_Y (Y)
\}-c. \Label{2}
\end{align}
\end{lemma}

\begin{remark}
The paper \cite[Theorem 1]{KV13} gives a similar statement with slightly different terminology.
To readers' convenience, we give its proof in Appendix \ref{A10}.
\end{remark}


\subsubsection{Conditional additive case}
Now, we proceed to the conditional additive case given in \eqref{ca1},
in which, ${\cal Y}$ is given as ${\cal X} \times {\cal Z}$.
Applying \eqref{2} to the conditional additive case, we obtain the following lemma. 
\begin{lemma}\Label{sbcl'}
The inequality 
\begin{align}
\Pjs(P_M, W_{X, Z|X}) \ge
P_M \times P_{XZ} 
\{ P_M(M) P_{X|Z}(X| Z)
\le c \frac{1}{|{\cal X}|}
\} - c \Label{e}
\end{align}
holds for any $c>0$.

\end{lemma}
\begin{proof}
We choose $Q_Y$ as
\begin{align*}
Q_Y(y) = Q_{XZ}(x, z) = \frac{1}{|{\cal X}|}P_Z(z)
\end{align*}
to \eqref{2}. Then, the first term of the right hand side of (\ref{e}) is
\begin{align*}
&\sum_m P_M(m)W_{Y|X=\san{e}(m)} 
\{ P_M(m)W_{Y|X=\san{e}(m)} (Y)
\le c Q_Y (Y)
\}\\
=&\sum_m P_M(m)W_{XZ|X=\san{e}(m)} 
\{ P_M(m)W_{XZ|X} (X, Z|\san{e}(m))
\le c \frac{1}{|{\cal X}|}P_Z (Z)
\}\\
=&\sum_m P_M(m)P_{XZ} 
\{ P_M(m)P_{XZ} (X-\san{e}(m), Z)
\le c \frac{1}{|{\cal X}|}P_Z (Z)
\}\\
=& P_M \times P_{XZ} 
\{ P_M(M)P_{XZ}(X, Z)
\le c \frac{1}{|{\cal X}|}P_Z (Z)
\}\\
=&
 P_M \times P_{XZ} 
\{ P_M(M) P_{X|Z}(X| Z)
\le c \frac{1}{|{\cal X}|}
\}.
\end{align*}
So, we obtain (\ref{e}). 
\end{proof}

\section{$n$-fold Markovian conditional additive channel}\Label{S5}
\subsection{Formulation for general case}\Label{s51}
Firstly, we give general notations for channel coding when the message obeys Markovian process. 
The formulation presented in this subsection will be applied even to the next section.
We assume that the set of messages is ${\cal M}^k$. 
Then, we assume that
the message $M^k=(M_1, \ldots, M_k)\in {\cal M}^k$ 
is subject to the Markov process with 
the transition matrix $\{W_s(m|m')\}_{m, m' \in {\cal M}}$. 
We denote the distribution for $M^k$ by $P_{M^k}$. 

Now, we consider very general sequence of channels with 
the input alphabet ${\cal X}^n$ and the output alphabet ${\cal Y}^n$. 
In this case, the transition matrix as 
$\{W_{Y^n| X^n}(y^n| x^n)\}_{x^n \in {\cal X}^n, y^n \in {\cal Y}^n}$. 
Then, a channel code $\phi = (\san{e}, \san{d})$ consists of one encoder 
$\san{e}: {\cal M}^k \to {\cal X}^n$ and
one decoder $\san{d}:{\cal Y}^n \to {\cal M}^k$. 
Then, the average decoding error probability is defined by
\begin{eqnarray}
\Pj [\phi|k, n|W_s, W_{Y^n|X^n}] := \sum_{m^k \in {\cal M}^k}
P_{M^k}(m^k) W_{Y^n| X^n}
(\{y^n:\san{d}(y^n) \neq m^k \}|\san{e}(m^k)). 
\end{eqnarray}
For notational convenience, we introduce the error probability 
under the above condition:
\begin{eqnarray}
\Pj(k, n|W_s, W_{Y^n|X^n}) := \inf_{\phi} \Pj[\phi|k, n|W_s, W_{Y^n|X^n}]. 
\end{eqnarray}
When there is no possibility for confusion, we simplify it to $\Pj(k, n)$. 
Instead of evaluating the error probability $\Pj(n, k)$ for given $n, k$, 
we are also interested in evaluating 
\begin{eqnarray}
\K(n, \varepsilon|W_s, W_{Y^n|X^n}) := \sup\left\{ k : \Pj(n, k|W_s, W_{Y^n|X^n}) \le \varepsilon \right\}
\end{eqnarray}
for given $0 \le \varepsilon \le 1$. 
\subsection{Formulation for Markovian conditional additive channel}\Label{S52}
In this section, we address an $n$-fold Markovian conditional additive channel \cite{HW}. 
That is, we consider the case when the joint distribution for the additive noise obeys the Markov process. 
To formulate our channel, we prepare notations. 
Consider the joint Markovian process on ${\cal X}\times {\cal Z}$. 
That is, the random variables $X^n=(X_1, \ldots, X_n)\in {\cal X}^n$ 
and $Z^n=(Z_1, \ldots, Z_n)\in {\cal Z}^n$ 
are assumed to be subject to 
the joint Markovian process defined by 
the transition matrix $\{W_c(x, z|x', z')\}_{x, x'\in {\cal X}, z, z' \in {\cal Z}}$. 
We denote the joint distribution for $X^n$ and $Z^n$ by $P_{X^n, Z^n}$. 
Now, we assume that ${\cal X}$ is a module, and 
consider the channel with
the input alphabet ${\cal X}^n$
and the output alphabet $({\cal X} \times {\cal Z})^n$. 
The transition matrix for the channel 
$W_{X^n, Z^n| {X^n}'}$ is given as
\begin{align}
W_{X^n, Z^n| {X^n}'}(x^n, z^n| {x^n}')=
P_{X^n, Z^n}(x^n-{x^n}', z^n )
\end{align}
for $z^n \in {\cal Z}^n$ and $x^n, {x^n}'\in {\cal X}^n$. 
Also, we denote $\log |{\cal X}|$ by $R$. 
In this case, we denote 
the average error probability $\Pj [\phi| k, n|W_s, W_{X^n, Z^n|X^n}]$ and
the minimum average error probability $\Pj (k, n|W_s, W_{X^n, Z^n|X^n})$
by $\rom{P_{jca}}[\phi| k, n|W_s, W_c]$ and
$\rom{P_{jca}}(k, n|W_s, W_c)$, respectively. 
Then, we denote the maximum size 
$\K(n, \varepsilon|W_s, W_{Y^n|X^n})$
by $\rom{K_{ca}}(n, \varepsilon|W_s, W_c)$. 
When we have no possibility for confusion, we simplify them to
by $\rom{P_{jca}}[\phi| k, n]$, $\rom{P_{jca}}(k, n)$, 
and $\rom{K_{ca}}(n, \varepsilon)$, respectively. 

In the following discussion, we assume the non-hidden condition for
the joint Markovian process described by 
the transition matrix $\{W_c(x, z|x', z')\}_{x, x'\in {\cal X}, z, z' \in {\cal Z}}$. 
Under the non-hidden condition, the paper \cite{HW} shows the single-letterized channel capacity 
to be $C:=\log|{\cal X}| - H^{W_c}(X|Z)$.
Among author's knowledge, the class of channels satisfying the non-hidden condition is the largest class of channels whose channel capacity is known. 
When ${\cal Z}$ is singleton and the channel is the noiseless channel given by identity transition matrix $I$, 
our problem becomes the source coding with Markovian source. 
In this case, the memory size is equal to the cardinality $|{\cal X}|^k$, and
we simplify the smallest attainable decoding error probability $\rom{P_{jca}}(k, n|W_s, I_{X|X})$ to $\Ps(k, n| W_s)$. 

\subsection{Second order analysis}\Label{S53}
\begin{theorem}\Label{T10}
For any $ 0 < \varepsilon < 1 $, 
it holds that 
\begin{align}
\lim_{n \rightarrow \infty} 
\frac{\rom{K_{ca}}(n, \varepsilon) H^{W_{s}}(M) - nC}{\sqrt{n}} 
	= \sqrt{\frac{C}{H^{W_{s}}(M)}V^{W_{s}}(M)+V^{W_{c}}(X|Z)}\Phi^{-1}(\varepsilon). 
\end{align}
In other words, 
\begin{align}
\lim_{n \to \infty}
\rom{P_{jca}}
\left (
n \frac{C}{H^{W_{s}}(M)}
+ \sqrt{n} \frac{ R }{ H^{W_s}(M) }, 
n
\right )
=
\Phi_{\frac{C}{H^{W_{s}}(M)}V^{W_{s}}(M)+V^{W_{c}}(X|Z)}(R )
. \Label{T10,err}
\end{align}
\end{theorem}

Theorem \ref{T10} yields the following corollary. 
\begin{corollary}
For $0<\varepsilon<1$, we have
\begin{align}
\lim_{n \rightarrow \infty} \frac{\K(n, \varepsilon)}{n} 
	= \frac{C}{H^{W_{s}}(M)}.
\end{align}
\end{corollary}
\begin{proof}
It is sufficient to show
\begin{align}
\lim_{n \to \infty}
\rom{P_{jca}}
\left (k, n
\right )
=
\Phi_{\frac{C}{H^{W_{s}}(M)}V^{W_{s}}(M)+V^{W_{c}}(X|Z)}(R )
. \Label{T10,errB}
\end{align}
when $k$ is chosen as 
\begin{align}
kH^{W_{s}}(M) 
=& nC + \sqrt{n} R 
= n \log|{\cal X}| - n H^{W_c}(X|Z)
 + \sqrt{n} R.
\end{align} 
By choosing $c= e^{n^{1/4}}$, (\ref{a}) implies that
\begin{align}
 \Pj(k, n) \le
P_{M^k} \times P_{X^nZ^n} 
\{ -\log{P_{M^k}(M^k)} - \log{P_{X^n|Z^n}(X^n|Z^n)}
\ge n\log{|{\cal X}|} - n^{1/4}
\} + e^{-n^{1/4}}. \Label{so1}
\end{align}
Applying Proposition \ref{CLT} to the random variables 
$-\log{P_{M^k}(M^k)} $ and $- \log{P_{X^n|Z^n}(X^n|Z^n)}$,
we find that
\par\noindent
the random variable $\frac{1}{\sqrt{n}}(-\log{P_{M^k}(M^k)} - \log{P_{X^n|Z^n}(X^n|Z^n)}
- k H^{W_{s}}(M) - n H^{W_c}(X|Z) )$
converges to the Gaussian random variable with variance 
$ \frac{C}{H^{W_{s}}(M)}V^{W_{s}}(M)+V^{W_{c}}(X|Z)$.
Since $\frac{n^{1/4}}{\sqrt{n}} \to 0$ and
$\frac{1}{\sqrt{n}}
( n\log{|{\cal X}|} )
=\frac{1}{\sqrt{n}}
(kH^{W_{s}}(M) 
+ n H^{W_c}(X|Z)
 - \sqrt{n} R)$, 
we see that the RHS of \eqref{so1} goes to $\Phi_{\frac{C}{H^{W_{s}}(M)}V^{W_{s}}(M)+V^{W_{c}}(X|Z)}(R )$, which implies that
\begin{align}
\limsup_{n\rightarrow \infty} \Pj(k, n) \le 
1-\Phi_{\frac{C}{H^{W_{s}}(M)}V^{W_{s}}(M)+V^{W_{c}}(X|Z)}(-R )
=\Phi_{\frac{C}{H^{W_{s}}(M)}V^{W_{s}}(M)+V^{W_{c}}(X|Z)}(R ).
\Label{EH1}
\end{align}

By choosing $c= e^{-n^{1/4}}$, 
(\ref{e}) implies that
\begin{align}
 \Pj(k, n) \ge
P_M \times P_{XZ} 
\{ -\log{P_{M^n}(M^n)} - \log{P_{X^n|Z^n}(X^n|Z^n)}
\ge n\log{|{\cal X}|} +n^{1/4}
\} - e^{-n^{1/4}} \Label{Eq10}
\end{align}
Since $ e^{-n^{1/4}}\to 0$, 
the above application of Proposition \ref{CLT} implies
\begin{align}
\liminf_{n\rightarrow \infty} \Pj(k, n) \ge 
1-\Phi_{\frac{C}{H^{W_{s}}(M)}V^{W_{s}}(M)+V^{W_{c}}(X|Z)}(-R )
\Label{EH2}.
\end{align}
The combination of \eqref{EH1} and \eqref{EH2} implies \eqref{T10,errB}.
\end{proof}
Similar to the above two cases, we can recover the result of data compression with the second order regime.


\section{$n$-fold Discrete Memoryless Channel (DMC) case}\Label{S6}
\subsection{Formulation and notations}\Label{S61} 
In this section, we address the $n$-fold discrete memoryless channel 
with the input system ${\cal X}^n$ and the output system ${\cal Y}^n$
Hence, we adopt the same assumptions given in Section \ref{S5}
for the message source. 
The difference from Section \ref{S5} is the form of channel. 
Given a transition matrix $\{W_{Y|X}(y|x)\}_{x\in {\cal X}, y\in {\cal Y}}$, 
the transition matrix for the channel 
$W_{Y^n| {X^n}}$ is given as
\begin{align}
W_{Y^n|X^n}(y^n|x^n) := \Pi^n_{i=1}W_{Y|X}(y_i|x_i)
\end{align}
where $x^n=(x_1, \ldots, x_n) \in {\cal X}^n$
and $y^n=(y_1, \ldots, y_n) \in {\cal Y}^n$. 

In this case, we denote the average error probability $\Pj [\phi| k, n|W_s, W_{X^n, Z^n|X^n}]$ and
the minimum average error probability $\Pj (k, n|W_s, W_{X^n, Z^n|X^n})$
by $\rom{P_{jdm}}[\phi| k, n|W_s, W_{Y|X}]$ and
$\rom{P_{jdm}}(k, n|W_s, W_{Y|X})$, respectively. 
Then, we denote the maximum size 
$\K(n, \varepsilon|W_s, W_{Y^n|X^n})$
by $\rom{K_{jdm}} (n, \varepsilon|W_s, W_{Y|X})$. 
When we have no possibility for confusion, we simplify them to 
$\rom{P_{jdm}}[\phi| k, n]$, $\rom{P_{jdm}}(k, n)$, and $\rom{K_{jdm}}(n, \varepsilon)$, respectively. 

For the latter discussion, we prepare the mutual information as
\begin{align*}
I(P_X, W_{Y|X}):=&\sum_{x\in {\cal X}} {P}_X(x)\sum_{y}W_{Y|X}(y|x)\log \frac{W_{Y|X}(y|x)}{\bar{W}_{Y}(y)} \\
=& \sum_{x\in {\cal X}}{P}_X(x) D(W_{Y|X=x}\| \bar{W}_{Y}) ,
\end{align*}
where
$D(P\|Q):= \sum_{y \in {\cal Y}}P(y)\log \frac{P(y)}{Q(y)} $.
Then, we define its variance version as
\begin{align}
V(P_X, W_{Y|X}):=
 \sum_{x}{P}_X(x)\sum_{y}\bar{W}_{Y}(y) \left(\log \frac{W_{Y|X}(y|x)}{\bar{W}_{Y}(y)} - 
D(W_{Y|X=x}\| \bar{W}_{Y}) \right)^2
\end{align}
and we also define the channel capacity
$C:= \max_{P_X \in {\cal P(X)}} I(P_X, W_{Y|X})= \min_{Q}\max_{x\in {\cal X}} D(W_{Y|X=x}\|Q) $.
Also, we define the  maximum and minimum variances
\begin{align}
V^*_{+} (W_{Y|X}) &:= \max_{ P_X :I ( P_X, W_{Y|X} ) = C } V( P_X, W_{Y|X} ) \\
V^*_{-} (W_{Y|X}) &:= \min_{ P_X :I ( P_X, W_{Y|X} ) = C } V( P_X, W_{Y|X} ),
\end{align}
and the distribution achieving above maximum and minimum as
\begin{align}
P_{X}^+(x)
& =
\mathop{ { \rm argmax } }_{ P_X :I ( P_X, W_{Y|X} ) = C } V( P_X, W_{Y|X} ),\\
P_{X}^-(x)
& =
\mathop{ { \rm argmin } }_{ P_X :I ( P_X, W_{Y|X} ) = C } V( P_X, W_{Y|X} ).
\end{align}

\subsection{Second order analysis and comparison}\Label{S62} 
Using the switched Gaussian convolution distribution $\Psi
 \left[ \frac{C}{H^{W_{s}}(M)} V^{W_{s}} (M), V^*_{+} (W_{Y|X}), V^*_{-} (W_{Y|X})
 \right] $, we derive the second order coding rate in the following Theorem.

\begin{theorem}\Label{dis,so,th}
For any $ \varepsilon \in (0, 1) $, we have
\begin{align}
\lim_{n \rightarrow \infty}
\rom{P_{jdm}}\left(
\frac{ C }{ H^{W_{s}}(M) }n + \frac{ R }{ H^{W_{s}}(M) } \sqrt{n}
, n \right)
=\varepsilon(R), \Label{eq14B}
\end{align}
where
\begin{align}
\varepsilon(R)
:=
\Psi
\left[ \frac{C}{H^{W_{s}}(M)} V^{W_{s}} (M), V^*_{+} (W_{Y|X}), V^*_{-} (W_{Y|X})
\right](R).
\end{align}
In other words, we have
\begin{align}
\lim_{n \rightarrow \infty} \frac{\rom{K_{jdm}}(n, \varepsilon)H^{W_{s}}(M) - nC}{\sqrt{n}} 
	= 
\Psi
 \left[ \frac{C}{H^{W_{s}}(M)} V^{W_{s}} (M), V^*_{+} (W_{Y|X}), V^*_{-} (W_{Y|X})
 \right]^{-1}(\varepsilon).
\end{align}
\end{theorem}

The direct and converse parts will be shown in Subsections \ref{S64} and \ref{S65}.
The paper \cite{DAY} discussed the same problem when 
the message is subject to the independent and identical distribution and 
the relation $V^*_{+} (W_{Y|X})=V^*_{-} (W_{Y|X})$ holds.
When the condition $V^*_{+} (W_{Y|X})=V^*_{-} (W_{Y|X})$ holds,
$\Psi 
\Big[\frac{C}{H^{W_{s}}(M)} V^{W_{s}} (M), V^*_{+} (W_{Y|X}), V^*_{-} (W_{Y|X})\Big]^{-1}(\varepsilon)$
becomes $
\sqrt{\frac{C}{H^{W_{s}}(M)} V^{W_{s}} (M)+ V^*_{+}(W_{Y|X})}
{\Phi^{-1} (\varepsilon)}$.

When the message is subject to the independent and identical distribution,
as a simple generalization of the direct part of \cite{DAY}, 
Kostina-Verd\'{u} \cite{KV} showed the inequality
	\begin{align}
\lim_{n \rightarrow \infty}
\rom{P_{jdm}}\left(
\frac{ C }{ H^{W_{s}}(M) }n + \frac{ R }{ H^{W_{s}}(M) } \sqrt{n}
, n \right)
\le \varepsilon_{KV}(R),
	\end{align}
	where $\varepsilon_{KV}(R) $ is defined as
	\begin{align}
	\varepsilon_{KV}(R):=
	\left\{
	\begin{array}{ll}
\Phi_{\frac{C}{H^{W_{s}}(M)} V^{W_{s}} (M)+V^*_{-}(W_{Y|X})}(R)	& \hbox{ when }R \le 0 \\
\Phi_{\frac{C}{H^{W_{s}}(M)} V^{W_{s}} (M)+V^*_{+}(W_{Y|X})}(R)	& \hbox{ when }R >0 .
	\end{array}
	\right.\Label{err,kv}
	\end{align}
Hence, we call the bound $\varepsilon_{KV}(R) $
Kostina-Verd\'{u} bound even for a general Markovian source with a transition matrix $W_{s}$.	
As a comparison between our tight bound $	\varepsilon(R)$ and Kostina-Verd\'{u} bound $	\varepsilon_{KV}(R)$,
we obtain the following lemma.

\begin{lemma}\Label{LH77}
The ratio $\frac{\varepsilon_{KV}(R) }{\varepsilon(R)}$
is evaluated as
\begin{align}
1 \le \frac{\varepsilon_{KV}(R) }{\varepsilon(R)}
\le
\left\{
\begin{array}{ll}
2 & \hbox{ when } R < 0 \\
1/\Phi_{ \frac{ C }{ H^{W_s}(M) } V^{ W_s } (M) }(R)
& \hbox{ when } R \ge 0.
\end{array}
\right. \Label{UPP7}
\end{align}
The equality of the first inequality is attained if and only if 
$ V^*_{+} (W_{Y|X}) = V^*_{-} (W_{Y|X}) $ or 
$ V(W_s) = 0 $.
The equality of the second inequality is attained if and only if 
$ V^*_{+} (W_{Y|X}) $ and $V^*_{-} (W_{Y|X}) $ go to $+\infty$ and $0$, respectively.
\end{lemma}
This lemma shows that a gap between 
$ V^*_{+} (W_{Y|X}) $ and $ V^*_{-} (W_{Y|X}) $ 
produces a non-negligible effect for joint source-channel coding
when the source is non-uniform.
Fig. \ref{joint2-6} gives a numerical calculation of the ratio
$\frac{\varepsilon_{KV}(R) }{\varepsilon(R)}$.

\begin{figure}[htbp]
\begin{center}
\scalebox{0.7}{\includegraphics[scale=1.3]{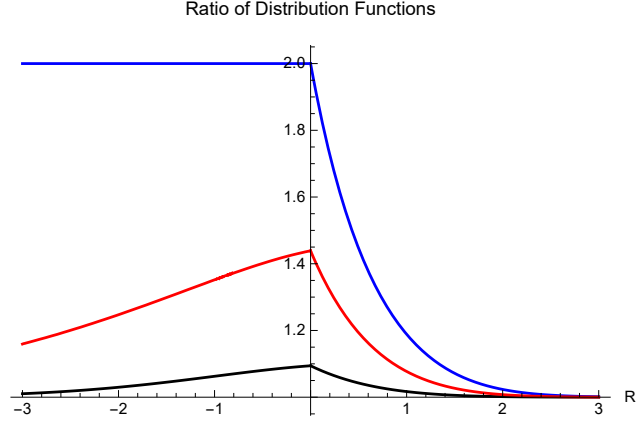}}
\end{center}
\caption{Graphs of the ratio $\frac{\varepsilon_{KV}(R) }{\varepsilon(R)}$
with $ \frac{ C }{ H^{W_s}(M) } V^{ W_s } (M) =1 $.
The origin is $(0,1)$.
Blue line expresses the upper bound given in \eqref{UPP7}.
Red line expresses the case with 
$ V^*_{-} (W_{Y|X}) =0.1$ and
$ V^*_{+} (W_{Y|X}) =10$. 
Black line expresses the case with 
$ V^*_{-} (W_{Y|X}) =0.5$ and
$ V^*_{+} (W_{Y|X}) =1.5$. }
\Label{joint2-6}
\end{figure}%

\begin{proof}
The property \eqref{12-24} implies the first inequality. 
The equality condition for the first inequality follows from the form of 
the switched Gaussian convolution distribution given in \eqref{12-24}.

To show the second inequality, we introduce the notation 
with variance $ v $ as:
\begin{align}
&\alpha[v](R)
:=
\int_{ R }^{ \infty }
	\Phi_{v}( R - x)
	\varphi_{\frac{ C }{ H^{W_s}(M) } V^{ W_s } (M)} (x)  dx\\
&\beta[v](R)
:=
\int_{ -\infty }^{ R }
\Phi_{v}( R - x)
\varphi_{\frac{ C }{ H^{W_s}(M) } V^{ W_s } (M)} (x)  dx.
\end{align}
For any $ R $, we find that 
$ \alpha[v](R) $ is 
monotonically increasing function 
of $ v $,
and 
$ \beta[v](R) $ 
is monotonically decreasing function 
of $ v $.
Additionally, we define
\begin{align}
&
\alpha_{max}(R)
:=
\lim_{v \to \infty}\alpha[v](R)
=
\frac{ 1 }{ 2 } \Phi_{ \frac{ C }{ H^{W_s}(M) } V^{ W_s } (M) }(-R)\\
&\alpha_{min}(R)
:=
\lim_{v \to 0}\alpha[v](R)
=
0\\
&\beta_{max}(R)
:=
\lim_{v \to 0}\beta[v](R)
=
\Phi_{ \frac{ C }{ H^{W_s}(M) } V^{ W_s } (M) }(R)\\
&
\beta_{min}(R)
:=
\lim_{v \to \infty}\beta[v](R)
=
\frac{ 1 }{ 2 } \Phi_{ \frac{ C }{ H^{W_s}(M) } V^{ W_s } (M) }(R).
\end{align}

For $ R < 0 $,
we have
\begin{align}
\frac{ \varepsilon_{KV}(R) }{ \varepsilon(R) } 
=& 
\frac{ \Psi \left [
	\frac{ C }{ H^{W_s}(M) } V^{ W_s } (M), V^*_{-} (W_{Y|X}), V^*_{-} (W_{Y|X})
	\right ](R) }
{ \Psi \left [
	\frac{ C }{ H^{W_s}(M) } V^{ W_s } (M), V^*_{+} (W_{Y|X}), V^*_{-} (W_{Y|X})
	\right ](R) } \nonumber \\
=&
\frac{ \alpha[V^*_{-} (W_{Y|X})](R) + \beta[V^*_{-} (W_{Y|X})](R) }
	{ \alpha[V^*_{-} (W_{Y|X})](R) + \beta[V^*_{+} (W_{Y|X})](R) } \nonumber \\
=&
1+\frac{ \beta[V^*_{-} (W_{Y|X})](R) - \beta[V^*_{+} (W_{Y|X})](R) }
	{ \alpha[V^*_{-} (W_{Y|X})](R) + \beta[V^*_{+} (W_{Y|X})](R) } \nonumber \\
\stackrel{(a)}{\le}  &
1 +
\frac{ \beta_{max}(R) - \beta_{min}(R) }
{ \alpha_{min}(R) + \beta_{min}(R) }
=2,\Label{dis,so,re1}
\end{align}
where $(a)$ follows from 
$ \beta[V^*_{+} (W_{Y|X})](R) \ge \beta_{min}(R)$, 
$ \beta[V^*_{-} (W_{Y|X})](R) \le \beta_{max}(R)$, 
and $ \alpha[V^*_{-} (W_{Y|X})](R) \ge \alpha_{min}(R)$.

For $ R \ge 0 $, we have
\begin{align}
&\frac{ \varepsilon_{KV}(R) }{ \varepsilon(R) } = 
\frac{ \Psi \left [
	\frac{ C }{ H^{W_s}(M) } V^{ W_s } (M), V^*_{+} (W_{Y|X}), V^*_{+} (W_{Y|X})
	\right ](R) }
{ \Psi \left [
	\frac{ C }{ H^{W_s}(M) } V^{ W_s } (M), V^*_{+} (W_{Y|X}), V^*_{-} (W_{Y|X})
	\right ](R) } \nonumber \\
=&
\frac{ \alpha[V^*_{+} (W_{Y|X})](R) + \beta[V^*_{+} (W_{Y|X})](R) }
{ \alpha[V^*_{-} (W_{Y|X})](R) + \beta[V^*_{+} (W_{Y|X})](R) } \nonumber \\
=&
1+ \frac{ \alpha[V^*_{+} (W_{Y|X})](R) - \alpha[V^*_{-} (W_{Y|X})](R) }
{ \alpha[V^*_{-} (W_{Y|X})](R) + \beta[V^*_{+} (W_{Y|X})](R) } \nonumber \\
\stackrel{(b)}{\le}  &
1 +
\frac{ \alpha_{max}(R) - \alpha_{min}(R) }
{ \alpha_{min}(R) + \beta_{min}(R) }
= 
1+ \frac{
\Phi_{ \frac{ C }{ H^{W_s}(M) } V^{ W_s } (M) }(-R)
}{
\Phi_{ \frac{ C }{ H^{W_s}(M) } V^{ W_s } (M) }(R)} \nonumber \\
=&
1/\Phi_{ \frac{ C }{ H^{W_s}(M) } V^{ W_s } (M) }(R),
\Label{dis,so,re4}
\end{align}
where $(b)$ follows from 
$ \beta[V^*_{+} (W_{Y|X})](R) \ge \beta_{min}(R)$, 
$ \alpha[V^*_{-} (W_{Y|X})](R) \le \alpha_{max}(R)$, 
and $ \alpha[V^*_{-} (W_{Y|X})](R) \ge \alpha_{min}(R)$.
The quality condition of the second inequality follows from the equality conditions of $(a)$ and $(b)$.
\end{proof}

\subsection{Direct part}\Label{S64}
To show the direct part of Theorem \ref{dis,so,th}, we invent a novel random coding method because the existing random coding method cannot attain the bound $\varepsilon(R)$.
To attain the bound $\varepsilon(R)$, we need to choose 
the distribution on ${\cal X}^n$ deciding the random coding
depending on the message to be sent.
Hence, we divide the set of messages into two sets,
and we decide our code depending on the set the message belongs to.
To realize this type code, we employ a code composed of two parts.
The first part informs which set the message belongs to.
The second part sends which element of the chosen set to be transmitted.
Using Proposition \ref{si,di,le}, we show that this code attains the bound
$\varepsilon(R)$.

\noindent{\it Step(0):} First, we prepare several notations, some of which are used throughout this proof including the converse part.
We simplify $W_{Y|X}(y|x)$ as $ W_{x}(y) $ and  $W_{Y^n|X^n}(y^n|x^n)$ as $ W_{x^n}^n(y^n) $. 
So, $W_{X^n}(Y^n)$ is a random variable on ${\cal X}^n \times {\cal Y}^n$.
We choose the integer $ k $ as
\begin{align}
k
&:= 
\frac{ C }{ H^{W_{s}}(M) }n + \frac{ R }{ H^{W_{s}}(M) } \sqrt{n}
\Label{eq14}.
\end{align}
Then, we define the following random variables.
\begin{align}
S(M^k)
&:=
 -\sqrt{n}\left (
 \frac{ - \log{P_{M^k} (M^k) } }{ n }
- \frac{ k }{ n } H^{W_{s}}(M)
\right ), \\
C(X^n,Y^n)
&:=
-\sqrt{n}
\left(
\frac{ 1 }{ n }
\log{\frac{W_{X^n}(Y^n)}{\bar{W}_{Y^n}(Y^n)}} - C
\right).
\end{align}

\noindent{\it Step (i):} In this step, 
we describe our code used in this proof.
This code consists of two parts as follows.
In the first part, the sender tells the receiver whether $ S(m^k) \le R $ or $ S(m^k) > R $.
In the second part, 
they communicate each other by using the code depending on the result of the first part.

Now, we give the first part, in which, the message size is $ 2 $. 
So, we use  only $ n^{1/4}$ transmission of the channel for the first part.
That is, the first is the code 
$ \phi_{n}^0 = (\san{e}_{n}^0, \san{d}_{n}^0)$ 
to tell whether $ S(m^k) \ge R $ or not.
Assume that ${\cal X}$ contains elements $0$ and $1$.
To give the first part,
we define the encoder $ \san{e}_{n}^0: \{0, 1 \} \to {\cal X}^{n} $ as
\begin{align*}
\san{e}_{n}^0(0)
&:= (0,0,\cdots,0)  \in {\cal X}^{n^{1/4}}\\
\san{e}_{n}^0(1)
&:= (1,1,\cdots,1) \in {\cal X}^{n^{1/4}}.
\end{align*}
The decoder 
$ \san{d}_{n}^0: {\cal Y}^{n^{1/4}} \to \{0, 1 \} $ is defined as
\begin{align*}
\san{d}_{n}^0(y) 
:=
\left\{
\begin{array}{ll}
0, & if \quad 
(W_{Y|X=0})^{\times n^{1/4}}(y) \ge (W_{Y|X=1})^{\times n^{1/4}}(y)\\
1, &
 if \quad
(W_{Y|X=0})^{\times n^{1/4}}(y) < (W_{Y|X=1})^{\times n^{1/4}}(y).
\end{array}
\right.
\end{align*}
Then, we denote the error probability of the code $ \phi_{n}^0 $
by $ \delta_{n} $, which is represented as
\begin{align}
\delta_{n} 
=&
(W_{Y|X=1})^{\times n^{1/4}}
\big\{ (W_{Y|X=0})^{\times n^{1/4}}(Y^{n^{1/4}}) \ge 
(W_{Y|X=1})^{\times n^{1/4}}(Y^{n^{1/4}}) \big\} 
\nonumber \\
& +
(W_{Y|X=0})^{\times n^{1/4}}
\big\{ (W_{Y|X=0})^{\times n^{1/4}}(Y^{n^{1/4}}) < (W_{Y|X=1})^{\times n^{1/4}}(Y^{n^{1/4}}) \big\}.
\Label{dis,so,di11}
\end{align}
Note that $ \delta_{n} \to 0 $ because $ n^{1/4} \to \infty $.

As the second part, we define the code to send the massage $ m^k $
based on the information transmitted in the first part.
We use $N$ transmissions of the channel in the second part, where $N=n-n^{1/4}$. 
Then, \eqref{eq14} implies that
\begin{align}
k
=
\frac{ C }{ H^{W_{s}}(M) }N + \frac{ R }{ H^{W_{s}}(M) } \sqrt{N} + o(\sqrt{N}).
\Label{dis,so,di9}
\end{align}
Using Proposition \ref{si,di,le},
we define the code $ \phi_N^+ := (\san{e}_N^+, \san{d}_N^-) $ 
so that 
\begin{align}
&\Pjs[\phi_N^+|P_{M^k |S(M^k) \le R} , W_{Y^{N}|X^{N}}] \nonumber \\
&\le
\left (
P_{M^k |S(M^k) \le R} \times (P_X^+)^{\times N} \times W_{Y^N|X^N}
\right )
\left\{ \log P_{M^k |S(M^k) \le R} (M^k) 
+ \log \frac{ W_{X^{N}}(Y^{N}) }{ W_{Y^{N}}(Y^{N}) }
\le \log c
\right\} + \frac{1}{c},\Label{dis,so,di4}
\end{align}
where $ P_{M^k |S(M^k) \le R}  $ is the conditional probability distribution 
of $ P_{M^k} $ under the condition of $ S(M^k) \le R $.
On the other hands, from Proposition \ref{si,di,le}, 
we define a code $ \phi_N^- = (\san{e}_N^-, \san{d}_N^-)$ 
so that 
\begin{align}
&\Pjs[\phi_N^-|P_{M^k |S(M^k) > R} , W_{Y^{N}|X^{N}}] \nonumber \\
&\le
\left (
P_{M^k |S(M^k) > R} \times (P_X^-)^{\times N} \times W_{Y^N|X^N}
\right )
\left\{ \log P_{M^k |S(M^k) > R} (M^k) 
+ \log \frac{ W_{X^{N}}(Y^{N}) }{ W_{Y^{N}}(Y^{N}) }
\le \log c
\right\} + \frac{1}{c},\Label{dis,so,di6}
\end{align}
where $ P_{M^k |S(M^k) > R} $ is the conditional probability distribution of
$ P_{M^k} $ under the condition of $ S(M^k) > R $.
In both cases, $c$ is chosen to be $e^{N^{1/4}}$.

Using the above preparation, we define the code
$ \phi_n := (\san{e}_n, \san{d}_n) $ 
for whole protocol as follows.
Then, for the encoder, 
we define $ \san{e}_n: {\cal M}^k \to
 {\cal X}^{\lceil n^{\frac{ 1 }{ 4 }} \rceil} \times {\cal X}^N $ as
\begin{align}
\san{e}_n(m^k)
:=
\left\{
\begin{array}{ll}
 \big( \san{e}_{n}^0 (0), \san{e}_N^+ (m^k) \big)  &
 {\rm when} \quad S(m^k) \le R  \\
 \big( \san{e}_{n}^0 (1), \san{e}_N^- (m^k) \big) & 
 {\rm when} \quad S(m^k) > R.
\end{array}
\right. 
\end{align}
Also we define the decoder 
$ \san{d}: {\cal X}^{\lceil N^{\frac{ 1 }{ 4 }} \rceil} \times {\cal X}^N \to
 {\cal M}^k $ as
\begin{align}
\san{d}(x_0, x_1)
:=
\left\{
\begin{array}{ll}
\san{d}_N^+ (x_1) &
{\rm when} \quad \san{d}_{n}^0 (x_0) = 0 \\
\san{d}_N^- (x_1) &
{\rm when} \quad \san{d}_{n}^0 (x_0) = 1.
\end{array}
\right. 
\end{align}

\noindent{\it Step (ii):} 
In this step, we will prove that
\begin{align}
&\Pjs[\phi|P_{M^k} , W_{Y^n|X^n}] \nonumber \\
&\le
P_{M^k} \{ S(M^k) \le R \} 
	\big(P_{M^k |S(M^k) \le R} \times (P_X^+)^{\times N} \times W_{Y^N|X^N} \big)
	\left\{ 
	S(M^k)
	-
	C(X^{N},Y^{N})
	\le R
	\right\} \nonumber \\
&\quad +
P_{M^k} \{ S(M^k) > R \}
	\big( P_{M^k |S(M^k) > R} \times (P_X^-)^{\times N} \times W_{Y^N|X^N} \big)
\left\{ 
S(M^k)
-
C(X^{N},Y^{N})
\le R
\right\} +o(1).\Label{eq11-1}
\end{align}

On the code $ \phi_n $, 
an error happens 
if 
an error occurs on the code $ \phi_{n}^0 $, 
or 
an error doesn't occur on the code $ \phi_{n}^0 $
and 
an error occurs on the code $ \phi_N^{\pm} $.
Since $ \delta_{n}\to 0$,
the error probability of the code $ \phi_n $, i.e.,  
$ \Pjs[\phi|P_{M^k} , W_{Y^n|X^n}] $, is evaluated as 
\begin{align}
&\Pjs[\phi|P_{M^k} , W_{Y^n|X^n}]\nonumber \\
&\le
P_{M^k} \{ S(M^k) \le R \} \Pjs[\phi^+|P_{M^k |S(M^k) \le R} , W_{Y^{N}|X^{N}}] 
+
P_{M^k} \{ S(M^k) > R \} \Pjs[\phi^-|P_{M^k |S(M^k) > R} , W_{Y^{N}|X^{N}}]
+o(1).
\Label{dis,so,di2}
\end{align}

When $ S(m^k) \le R $,
$ P_{M^k |S(M^k) \le R} (m^k) = \frac{ P_{M^k} (m^k) }{ P_{M^k} \{ S(M^k) \le R \} }$.
So, applying 
the central limit theorem for Markovian process (Proposition \ref{CLT}) to
random variable $ -\log P_{M^k} (M^k) $, 
we have 
\begin{align*}
P_{M^k} \{ S(M^k) \le R \} \to 
\Phi_{ \frac{ C }{ H^{W_s}(M) } V^{W_{s}} (M) } 
 (R)
\quad (N \to \infty),
\end{align*} 
which implies 
$ \log{P_{M^k|S(M^k) \le R} (M^k) }
=\log P_{M^k}(M^k) + o(\sqrt{N}). $
Since $ kH^{W_s}(M) = NC + \sqrt{N} R + o(\sqrt{N}) $
and $\frac{1}{\sqrt{N}}\log c \to 0$,
due to (\ref{dis,so,di9}), 
we can rewrite (\ref{dis,so,di4}) as:
\begin{align}
&\Pjs[\phi^+|P_{M^k |S(M^k) \le R} , W_{Y^{N}|X^{N}}] \nonumber \\
&\le
\big( P_{M^k |S(M^k) \le R} \times (P_X^+)^{\times n} \times W_{Y^N|X^N} \big)
\left\{ 
S(M^k)
- 
C(X^{N},Y^{N})
\le R
\right\}+o(1).\Label{dis,so,di5}
\end{align}

On the other hands, 
when $ S(m^k) > R $, 
we have $ P_{M^k |S(M^k) > R} (m^k) = \frac{ P_{M^k} (m^k) }
{ P_{M^k} \{ S(M^k) \le R \} }$. 
So, applying  
the central limit theorem for Markovian process to 
random variable $ -\log P_{M^k} (M^k) $,
we obtain 
\begin{align*}
P_{M^k} \{ S(m^k) \le R \} \to 
\Phi_{ \frac{ C }{ H^{W_s}(M) } V^{W_{s}} (M) } (R)
\quad (n \to \infty),
\end{align*} 
which implies
	$ \log{P_{M^k|S(M^k) > R} (m^k) }
	=
	\log P_{M^k}(m^k) + o(N). $
So, we can rewrite (\ref{dis,so,di6}) as:
\begin{align}
	&\Pjs[\phi^-|P_{M^k |S(M^k) \le R} , W_{Y^{N}|X^{N}}] \nonumber \\
	&\le
	\big( P_{M^k |S(M^k) \le R} \times (P_X^-)^{\times n} \times W_{Y|X} \big)
	\left\{  
	S(M^k)
	- 
	C(X^{N},Y^{N})
	\le R
	\right\}+o(1).\Label{dis,so,di7}
\end{align}
Combining (\ref{dis,so,di2}), (\ref{dis,so,di5}) and (\ref{dis,so,di7}),
we obtain \eqref{eq11-1}.

\noindent{\it Step (iii):} 
In this step, 
we will prove that
\begin{align}
\limsup_{n \to \infty}
\Pjs[\phi|P_{M^k} , W_{Y^n|X^n}] 
\le \varepsilon(R)
\Label{dis,so,di10},
\end{align}
which implies 
\begin{align}
\limsup_{n \to \infty}
\rom{P_{jdm}}(k, n)
\le \varepsilon(R)
\Label{dis,so,di10B}
\end{align}
for the integer $k$ given in \eqref{eq14}.

Applying the central limit theorem for Markovian process (Proposition \ref{CLT}), we find the following facts.
Under the distribution $P_{M^k} $, the random variable $ S(M^k) $ asymptotically obeys the Gaussian distribution with mean $ 0 $ and variance $ \frac{C}{H^{W_{s}}(M)} V^{W_{s}} (M) $.
Under the distribution $P_{M^k |S(M^k) \le R} \times (P_X^+)^{\times N} \times W_{Y^N|X^N} $, 
the random variable $ C(X^{N}, Y^{N}) $ asymptotically obeys the Gaussian distribution with mean $ 0 $ 
and variance $ V^*_{+} (W_{Y|X}) $.
Under the distribution $P_{M^k |S(M^k) \le R} \times (P_X^-)^{\times N} \times W_{Y^N|X^N} $, 
the random variable $ C(X^{N}, Y^{N}) $ asymptotically obeys the Gaussian distribution with mean $ 0 $ 
and variance $ V^*_{-} (W_{Y|X}) $.
Hence, taking the limit $N \to \infty$, 
we obtain
\begin{align}
&P_{M^k} \{ S(M^k) \le R \} 
\big( P_{M^k |S(M^k) \le R} \times (P_X^+)^{\times N} \times W_{Y^N|X^N} \big)
\left\{ 
S(M^k)
-C(X^{N},Y^{N})
\le R
\right\} \nonumber \\
&+
P_{M^k} \{ S(M^k) > R \}
\big( P_{M^k |S(M^k) > R} \times (P_X^-)^{\times n} \times W_{Y^N|X^N} \big)
\left\{ 
S(M^k)
-
C(X^{N},Y^{N})
\le R
\right\} \nonumber \\
\to&
\Phi_{ \frac{ C }{ H^{W_s}(M) } V^{W_{s}} (M) } 
(-R)
\frac{ 1 }{ 
\Phi_{ \frac{ C }{ H^{W_s}(M) } V^{W_{s}} (M) } 
(-R) }
\int_{ -\infty }^{ R } \varphi_{\frac{C}{H^{W_{s}}(M)} V^{W_{s}} (M)} (x)
\left( \int_{ -R + x }^{ \infty } \varphi_{V^*_{+} (W_{Y|X})} (y) dy \right) dx \nonumber \\
&+
\Phi_{ \frac{ C }{ H^{W_s}(M) } V^{W_{s}} (M) }
 (R)
\frac{ 1 }{ \Phi_{ \frac{ C }{ H^{W_s}(M) } V^{W_{s}} (M) } 
 (R)}
\int_{ R }^{ \infty } \varphi_{\frac{C}{H^{W_{s}}(M)} V^{W_{s}} (M)} (x)
\left( \int_{ -R + x }^{ \infty } \varphi_{V^*_{-} (W_{Y|X})} (y) dy \right) dx
\nonumber \\
=&
\int_{ -\infty }^{ R } \varphi_{\frac{C}{H^{W_{s}}(M)} V^{W_{s}} (M)} (x)
\left( \int_{ -\infty }^{ R-x } \varphi_{V^*_{+} (W_{Y|X})} (y) dy \right) dx
+
\int_{ R }^{ \infty } \varphi_{\frac{C}{H^{W_{s}}(M)} V^{W_{s}} (M)} (x)
\left( \int_{ -\infty }^{ R-x } \varphi_{V^*_{-} (W_{Y|X})} (y) dy \right) dx  \nonumber\\
=&
\Psi
\left[ \frac{C}{H^{W_{s}}(M)} V^{W_{s}} (M), V^*_{+} (W_{Y|X}), V^*_{-} (W_{Y|X})
\right](R)
= \varepsilon(R),
\end{align}
which implies (\ref{dis,so,di10}).

\subsection{Converse part}\Label{S65}
To show the converse part, we apply (\ref{2}) of Lemma \ref{L-11} to the case with the distribution $Q^n_U $ given in Step (i),
which can be regarded as an extension of the idea of the paper \cite{6} to the joint scheme.
Then, we apply the central limit theorem for Markovian process 
(Proposition \ref{CLT})
to the two random variables related to the dispersions of channel and source.
Since we treat two Gaussian random variables,
the asymptotic error probability is lower bounded by the convolution of 
two Gaussian distributions.
However, since the variance of the dispersions of channel is not unique, in general,
we need to take the minimum for the Gaussian distribution function.
Hence, 
the asymptotic error probability is lower bounded by the switched Gaussian convolution distribution.

\noindent{\it Step (i):} In this step, to show the converse part, we prepare several notations. 
We choose the message block length $ k $ so that
\begin{align}
kH^{W_{s}}(M) = nC + \sqrt{n}R + n^{1/4}.
\end{align}
We denote that $ x^n := \rom{ e } ( m^k ) $. 
We focus on the set $T_n$ of empirical distributions with $n$ channel inputs. 
Its cardinality $|T_n|$ is evaluated as $|T_n| \le (n+1)^{|{\cal X}|}$. 
And in this proof, we use the distribution 
\begin{align}
Q^n_U :=
\sum_{P \in T_n}
	\frac{1}{|T_n|+1}(W_P)^{\times n}
	+\frac{1}{|T_n|+1}Q^{\times n}_M , 
\end{align}
where
\begin{align}
Q_M := \mathop{ { \rm argmin } }_{Q} \max_{x}D(W_x \| Q). 
\end{align}
We also define the sets
\begin{align}
\nu_\xi
:=&\{P|I(P, W_{Y|X}) \ge C-\xi \}, \\
\Omega_n
:=&\{m^k \in {\cal M} ^k | {\rm ep}( \san{ e } ( m^k ) ) \in \nu_\xi \}, \Label{ep}\\
\pi_{ n, J, i }
:=&
\left \{
	m^k \in { \cal M }^k 
	\left|
	\frac{ i }{ J }
	\le
	S(m^k)
		 \le 
	 \frac{ i+1 }{ J }
	\right. 
\right \}, 
\end{align}
where $ {\rm ep}(\san{ e } ( m^k )) $ of \eqref{ep} is empirical distribution function
of $ \san{ e } ( m^k ) \in {\cal X}^n $.

\noindent{\it Step (ii):}\quad
We set the real number $c$ to be $e^{ - n^{ \frac{ 1 }{ 4 } }}$.
Since $ \log c= nC + \sqrt{n}R -kH^{W_{s}}(M) $, by substituting $ Q_{ Y } = Q^n_U $, 
(\ref{2}) of Lemma \ref{L-11} implies that 
\begin{align*}
&\Pj[\phi|k, n] \\
\ge&
\sum_{m^k} 
P_{M^k} (m^k)W_{\san{e}(m^k)}^n
\left \{
S(M^k)
	+ \sqrt{n} \left (
\frac{ 1 }{ n }
\log{\frac{W_{\san{e}(m^k)}^n(Y^n)} { Q^n_U ( Y ) } }- C
				\right )
\le R
\right \}
- e^{ - n^{ \frac{ 1 }{ 4 } } }.
\end{align*}

For arbitrary $ L >0 $, the first term of right hand side is evaluated as 
\begin{align}
&\sum_{m^k} 
P_{M^k} (m^k)W_{\san{e}(m^k)}^n
\left \{
S(M^k)
	+ \sqrt{n} \left (
	\frac{ 1 }{ n }
				\log{\frac{W_{\san{e}(m^k)}^n(Y^n)} { Q^n_U ( Y^n ) } }- C
				\right )
\le R
\right \} \nonumber \\
\ge&
\sum_{ i= -LJ } ^{ LJ-1 }
	\sum_{m^k \in \pi_{ n, J, i }} 
	P_{M^k} (m^k) W_{\san{e}(m^k)}^n
	\left \{
	S(m^k)
	+ \sqrt{n} \left (
	\frac{ 1 }{ n }
				\log{\frac{W_{\san{e}(m^k)}^n(Y^n)} { Q^n_U ( Y^n ) } }- C
				\right )
\le R
\right \} \nonumber \\
\ge&
\sum_{ i= -LJ } ^{ LJ-1 }
	\sum_{m^k \in \pi_{ n, J, i }} 
	P_{M^k} (m^k) W_{\san{e}(m^k)}^n
	\left \{
	\frac{ i+1 }{ J }
	+ \sqrt{n} \left (
	\frac{ 1 }{ n }	\log{\frac{W_{\san{e}(m^k)}^n(Y^n)} { Q^n_U ( Y^n ) } }- C
				\right )
\le R
\right \} \nonumber \\
=&
\sum_{ i= -LJ } ^{ LJ-1 }
	\sum_{m^k \in \pi_{ n, J, i } \cap \Omega_n} 
	P_{M^k} (m^k) W_{\san{e}(m^k)}^n
	\left \{
	\frac{ i+1 }{ J }
	+ \sqrt{n} \left (
	\frac{ 1 }{ n } \log{\frac{W_{\san{e}(m^k)}^n(Y^n)} { Q^n_U ( Y^n ) } }- C
				\right )
\le R
\right \} \nonumber \\
&+\sum_{ i= -LJ } ^{ LJ-1 }
	\sum_{m^k \in \pi_{ n, J, i } \cap \Omega_n^c } 
	P_{M^k} (m^k) W_{\san{e}(m^k)}^n
	\left \{
	\frac{ i+1 }{ J }
	+ \sqrt{n} \left (
	\frac{ 1 }{ n }	\log{\frac{W_{\san{e}(m^k)}^n(Y^n)} { Q^n_U ( Y^n ) } }- C
				\right )
\le R
\right \} \nonumber \\
\ge&
\sum_{ i= -LJ } ^{ LJ-1 }
	\sum_{m^k \in \pi_{ n, J, i } \cap \Omega_n} 
	P_{M^k} (m^k) W_{\san{e}(m^k)}^n
	\left \{
	\frac{ i+1 }{ J }
	+ \sqrt{n} \left (
	\frac{ 1 }{ n }	\log{\frac{W_{\san{e}(m^k)}^n(Y^n)} { (Q_M )^{ \times n }( Y^n ) } }
				+ \frac{ 1 }{ n } \log{(|T_n|+1)} - C
				\right )
\le R
\right \} \nonumber \\
&+\sum_{ i= -LJ } ^{ LJ-1 }
	\sum_{m^k \in \pi_{ n, J, i } \cap \Omega_n^c } 
	P_{M^k} (m^k) W_{\san{e}(m^k)}^n
	\left \{
	\frac{ i+1 }{ J }
	+ \sqrt{n} \left (
	\frac{ 1 }{ n }	\log{\frac{W_{\san{e}(m^k)}^n(Y^n)} { W_{ {\rm ep} (\san{e}(m^k)) }^{ \times n }( Y^n ) } }
				+ \frac{ 1 }{ n } \log{(|T_n|+1)} - C
				\right )
\le R
\right \}. \Label{dis,so,co2}
\end{align}

\noindent{\it Step (iii):} 
For the second term of (\ref{dis,so,co2}), we will show the following fact: 
Given an arbitrary small real number $\delta>0$, 
there exists a sufficiently large $n_1$ such that
\begin{align}
&
	W_{\san{e}(m^k)}^n 
	\left \{
	\frac{ i+1 }{ J }
	+ \sqrt{n} \left (
	\frac{ 1 }{ n }
		\log{\frac{W_{\san{e}(m^k)}^n(Y^n)} { W_{ {\rm ep} (\san{e}(m^k)) }^{ \times n }( Y^n ) } }
				+ \frac{ 1 }{ n } \log{(|T_n|+1)} - C
				\right )
\le R
\right \} \nonumber \\
\ge& 
1 - \delta \Label{ dis, so, pr1 }, 
\end{align} 
for $n \ge n_1$ and $m^k \in \pi_{n,J, i}\cap \Omega_n^c $. 

When $m^k \in \Omega_n^c$, 
\begin{align*}
\V_{ W_{Y^n|X^n = \san{e} (m^k) } }&
\left[\frac{ 1}{\sqrt { n } }
\left(
	\log{\frac{W_{\san{e}(m^k)}^n(Y^n)}{(W_{ {\rm ep} (\san{e}(m^k))})^{\times n}(Y^n)}} 
	+ \log{(|T_n|+1)} - nC
	\right)
	\right]\\
&=
\V _{ {\rm ep} (\san{e}(m^k)), W } < \max_{ P_{ X } } \V _{ P_{ X }, W }  	,
\end{align*}
\begin{align*}
\rom{E}_{ W_{Y^n|X^n = \san{e} (m^k) } }&
\left[\frac{ 1}{\sqrt { n } }
\left(
	\log{\frac{W_{\san{e}(m^k)}^n(Y^n)}{(W_{ {\rm ep} (\san{e}(m^k))})^{\times n}(Y^n)}} 
	+ \log{(|T_n|+1)} - nC
	\right)\right]\\
&=
\frac{ 1}{\sqrt { n } }
	(nI( {\rm ep} (\san{e}(m^k)), W_{Y|X}) + \log{(|T_n|+1)} - nC) \nonumber \\
&\le
\frac{ \log{(|T_n|+1)}}{ \sqrt { n } } - \xi \sqrt { n } , 
\end{align*} 
where $\rom { E } _{ P } $ and $\V _{ P } $ denote the expectation and the variance under the distribution $P$. 
Thus, when $ m^k \in \pi_{ n, J, i } \cap \Omega_n^c $, by using Chebyshev inequality, we obtain
\begin{align}
&W_{\san{e}(m^k)}^n
	\left \{
	\frac{ i+1 }{ J }
	+ \sqrt{n} \left (
 \frac{ 1 }{ n }
				\log{\frac{W_{\san{e}(m^k)}^n(Y^n)} { W_{ {\rm ep} (\san{e}(m^k)) }^{ \times n }( Y^n ) } }
				+ \frac{ 1 }{ n } \log{(|T_n|+1)} - C
				\right )
\le R
\right \} \nonumber \\
\ge &
	1 
	- \frac{ \V _{ {\rm ep} (\san{e}(m^k)), W } }
		{
		\left[
		 R 
			-  \frac{ i }{ J }
				- \frac{ 1}{\sqrt { n } } (nI( {\rm ep} (\san{e}(m^k)), W_{Y|X}) + \log{(|T_n|+1)} - nC) 
				\right]^2
				} .
\end{align} 
For sufficiently large $ n $, we have
\begin{align}
&
	W_{\san{e}(m^k)}^n
	\left \{
	\frac{ i+1 }{ J }
	+ \sqrt{n} \left (
 \frac{ 1 }{ n }
				\log{\frac{W_{\san{e} (m^k)}(Y^n)} { W_{ {\rm ep} (\san{e}(m^k)) }^{ \times n }( Y^n ) } }
				+ \frac{ 1 }{ n } \log{(|T_n|+1)} - C
				\right )
\le R
\right \} \nonumber \\
\ge&
	1 
	- \frac{ \max_{ P_{ X } } \V _{ P_{ X }, W } }
		{
		\left[
		 R 
			-  \frac{ i+1 }{ J }
				- \frac{ \log{(|T_n|+1)}}{ \sqrt { n } } + \xi \sqrt { n } 
				\right]^2
				} . 
\end{align}
Since the value
\begin{align*}
1 
	- \frac{ \max_{ P_{ X } } \V _{ P_{ X }, W } }
		{
		\left[
		 R 
			-  \frac{ i+1 }{ J }
				- \frac{ \log{(|T_n|+1)}}{ \sqrt { n } } + \xi \sqrt { n } 
				\right]^2
				} 
\end{align*}
asymptotically goes to $ 1 $, we obtain (\ref{ dis, so, pr1 }). \\

\noindent{\it Step (iv):} For the second term of (\ref{dis,so,co2}), we will show the following fact:\\
Given an arbitrary small real number $\delta>0$, 
there exists a sufficiently large $n_2$ such that
\begin{align}
&W_{\san{e}(m^k)}^n 
	\left \{
	\frac{ i+1 }{ J }
	+ \sqrt{n} \left (
 \frac{ 1 }{ n }
				\log{\frac{W_{\san{e}(m^k)}^n(Y^n)} { (Q_M )^{ \times n }( Y^n ) } }
				+ \frac{ 1 }{ n } \log{(|T_n|+1)} - C
				\right )
\le R
\right \} \nonumber \\
\ge&
\left \{
	\begin{array}{ll}
	\Phi_{ V^*_{+} (W_{Y|X}) }( R - \frac{ i+1 }{ J } )
	- \delta & \hbox { when } R \ge \frac{ i+1 }{ J }\\
	\Phi_{ V^*_{-} (W_{Y|X}) }( R - \frac{ i+1 }{ J })- \delta
	&\hbox{ when } R < \frac{ i+1 }{ J },
	\end{array}
\right.
\Label{ dis, so, pr2 }
\end{align}
for $n \ge n_2$ and $m^k \in \Omega_n$. 

Now, to evaluate the variance of some random variable later, we define the quantity
\begin{align}
\V' _{ P, W } 
:=
\rom{E} _{ P } \rom { E } _{ W_x }  
	\left(
		\log \frac{{ W_x} }{ Q_M } 
		-D(W_x \| Q_M)
			\right)^2. 
\end{align} 
When $m^k \in \Omega_n $, the inequality
\begin{align}
&W_{\san{e}(m^k)}^n
	\left \{
	\frac{ i+1 }{ J }
	+ \sqrt{n} \left (
				\log{\frac{W_{\san{e}(m^k)}^n(Y^n)} { (Q_M )^{ \times n }( Y^n ) } }
				+ \frac{ 1 }{ n } \log{(|T_n|+1)} - C
				\right)
\le R
\right \} \nonumber \\
\ge&
W_{\san{e}(m^k)}^n
	\left \{
	\frac{ i+1 }{ J }
	+ \sqrt{n} \left (
				\log{\frac{W_{\san{e}(m^k)}^n(Y^n)} { (Q_M )^{ \times n }( Y^n ) } }
				+ \frac{ 1 }{ n }  \log{(|T_n|+1)} - I( {\rm ep} (\san{e}(m^k)), W_{Y|X}) 
				\right )
\le R
\right \} 
\end{align} 
holds. Since the random variable
\begin{align}
\log{\frac{W_{\san{e}(m^k)}^n(Y^n)} { (Q_M )^{ \times n }( Y^n ) } }
=
\sum_i
	\log { \frac{ W_{x_i} ( Y_i ) }{ Q_M ( Y_i ) } }
\end{align}
 has the variance $ n V'_{ {\rm ep} (\san{e}(m^k)), W } $, applying the central limit theorem, we have 
\begin{align}
&W_{\san{e}(m^k)}^n
	\left \{
	\frac{ i+1 }{ J }
	+ \sqrt{n} \left (
				\log{\frac{W_{\san{e}(m^k)}^n(Y^n)} { (Q_M )^{ \times n }( Y^n ) } }
				+ \frac{ 1 }{ n }  \log{(|T_n|+1)} - I( {\rm ep} (\san{e}(m^k)), W_{Y|X}) 
				\right )
\le R
\right \} \nonumber \\
&\ge
\Phi
\left (\frac{ R - \frac{ i+1 }{ J }}{ \sqrt{V'_{ {\rm ep} (\san{e}(m^k)), W }} } \right ) - \delta
=
\Phi_{ V'_{ {\rm ep} (\san{e}(m^k)), W } }
\left ( R - \frac{ i+1 }{ J } \right ) - \delta,
\end{align}
for sufficiently large $ n $.
Because $ \Phi(\cdot) $ is a monotonicity increasing function and
the inequalities 
\begin{align}
V^*_{-} (W_{Y|X})
\le
V'_{ {\rm ep} (\san{e}(m^k)), W }
\le
V^*_{+} (W_{Y|X})
\end{align}
holds, 
the condition $ R - \frac{ i+1 }{ J } \ge 0  $  implies
\begin{align}
\Phi_{ V'_{ {\rm ep} (\san{e}(m^k)), W } }\left (
R - \frac{ i+1 }{ J } 
\right )
\ge
\Phi_{ V^*_{+} (W_{Y|X}) }
\left ( R - \frac{ i+1 }{ J }  \right ),
\end{align}
and the other condition $ R - \frac{ i+1 }{ J } < 0  $ implies
\begin{align}
\Phi_{ V'_{ {\rm ep} (\san{e}(m^k)), W } }
\left (  R - \frac{ i+1 }{ J }  \right )
\ge
\Phi_{ V^*_{-} (W_{Y|X}) }
\left (
R - \frac{ i+1 }{ J } 
\right ).
\end{align}
Hence, we obtain (\ref{ dis, so, pr2 }). 

\noindent{\it Step (v) :} We will show the following fact: 
Given an arbitrary small real number $\delta>0$, 
there exists a sufficiently large $n_3$ such that
\begin{align}
&\Pj[\phi|k, n]\nonumber \\
\ge &
\sum_{ i= -LJ } ^{ i_0 } 
P_{M^k} \left \{
\frac{ i }{ J } \le S(m^k) \le \frac{ i+1 }{ J }
\right \}
\Phi \left(
\frac{ R - \frac{ i+1 }{ J } }{ \sqrt { V^*_{+} (W_{Y|X}) } } 
\right)\nonumber \\
&+
\sum_{ i= i_0 } ^{ LJ-1 } 
P_{M^k} \left \{
\frac{ i }{ J } \le S(m^k) \le \frac{ i+1 }{ J }
\right \}
\Phi_{  V^*_{-} (W_{Y|X})  } 
 \left(  R - \frac{ i+1 }{ J } \right) - \delta,
\end{align}
where $ i_0 := \max \{ i \in \mathbb{Z} | \frac{ i+1 }{ J } \le R \} $,
for $n \ge n_3$  and $ m^k \in \Omega_n $.

Combining (\ref{ dis, so, pr1 }) and (\ref { dis, so, pr2 }), for sufficiently large $ n $, we obtain
\begin{align*}
&\Pj[\phi|k, n] \\
\ge&
\sum_{ i= -LJ } ^{ LJ-1 }
	\sum_{m^k \in \pi_{ n, J, i } \cap \Omega_n} 
	P_{M^k} (m^k) W_{\san{e}(m^k)}^n
	\left \{
	\frac{ i+1 }{ J }
	+ \sqrt{n} \left (
 \frac{ 1 }{ n }
				\log{\frac{W_{\san{e}(m^k)}^n(Y^n)} { (Q_M )^{ \times n }( Y^n ) } }
				+ \frac{ 1 }{ n }\log{(|T_n|+1)} - C
				\right )
\le R
\right \} \nonumber \\
&+\sum_{ i= -LJ } ^{ LJ-1 }
	\sum_{m^k \in \pi_{ n, J, i } \cap \Omega_n^c } 
	P_{M^k} (m^k) W_{\san{e}(m^k)}^n
	\left \{
	\frac{ i+1 }{ J }
	+ \sqrt{n} \left (
 \frac{ 1 }{ n }
				\log{\frac{W_{\san{e}(m^k)}^n(Y^n)} { W_{ {\rm ep} (\san{e}(m^k)) }^{ \times n }( Y^n ) } }
				+ \frac{ 1 }{ n } \log{(|T_n|+1)} - C
				\right )
\le R
\right \} \\
\ge &
\sum_{ i= -LJ } ^{ LJ-1 }
	\sum_{m^k \in \pi_{ n, J, i } \cap \Omega_n} 
	P_{M^k} (m^k) W_{\san{e}(m^k)}^n
	\left \{
	\frac{ i+1 }{ J }
	+ \sqrt{n} \left (
	\frac{ 1 }{ n }
	\log{\frac{W_{\san{e}(m^k)}^n(Y^n)} { (Q_M )^{ \times n }( Y^n ) } }
	+ \frac{ 1 }{ n }\log{(|T_n|+1)} - C
	\right )
	\le R
	\right \}		
 \nonumber\\
&+\sum_{ i= -LJ } ^{ LJ-1 }
	\sum_{m^k \in \pi_{ n, J, i } \cap \Omega_n^c } 
	P_{M^k} (m^k) \cdot 1 - \delta \\
\ge &
\sum_{ i= -LJ } ^{ LJ-1 }
	\sum_{m^k \in \pi_{ n, J, i }} 
	P_{M^k} (m^k) W_{\san{e}(m^k)}^n
	\left \{
	\frac{ i+1 }{ J }
	+ \sqrt{n} \left (
	\frac{ 1 }{ n }
	\log{\frac{W_{\san{e}(m^k)}^n(Y^n)} { (Q_M )^{ \times n }( Y^n ) } }
	+ \frac{ 1 }{ n }\log{(|T_n|+1)} - C
	\right )
	\le R
	\right \} - \delta \\
\ge&
\sum_{ i= -LJ } ^{ i_0 } 
P_{M^k} \left \{
\frac{ i }{ J } \le S(m^k) \le \frac{ i+1 }{ J }
\right \}
\Phi_{  V^*_{+} (W_{Y|X}) } 
 \left(
 R - \frac{ i+1 }{ J } 
 \right)\nonumber \\
&+
\sum_{ i= i_0 } ^{ LJ-1 } 
P_{M^k} \left \{
\frac{ i }{ J } \le S(m^k) \le \frac{ i+1 }{ J }
\right \}
\Phi_{  V^*_{-} (W_{Y|X})  } 
 \left(
 R - \frac{ i+1 }{ J } 
 \right) - \delta.
\end{align*}

\noindent{\it Step (vi):} We will show the following fact: Given an arbitrary small real number $\delta'>0$, 
there exist sufficiently large numbers $n_4, L$, and $J$ such that
\begin{align}
&
\sum_{ i= -LJ } ^{ i_0 } 
P_{M^k} \left \{
\frac{ i }{ J } \le S(m^k) \le \frac{ i+1 }{ J }
\right \}
\Phi_{ V^*_{-} (W_{Y|X})  } 
 \left(
 R - \frac{ i+1 }{ J } 
 \right)\nonumber \\
& +
\sum_{ i= i_0 } ^{ LJ-1 } 
P_{M^k} \left \{
\frac{ i }{ J } \le S(m^k) \le \frac{ i+1 }{ J }
\right \}
\Phi_{V^*_{+} (W_{Y|X})  } 
 \left(
 R - \frac{ i+1 }{ J } 
 \right) \nonumber \\
\ge & 
\varepsilon(R) - \delta' \Label{dis,so,st6},
\end{align}
for $n \ge n_4$. 

From the central limit theorem for Markov sequence 
(Proposition \ref{CLT}), random variable
$ S(M^k) $
asymptotically obeys Gaussian distribution with mean $ 0 $ and variance $ \frac{C}{H^{W_{s}}(M)}V^{ W_s }(M) $ i.e.,
\begin{align}
P_{M^k} \left \{
		\frac{ i }{ J } \le S(m^k) \le \frac{ i+1 }{ J }
	\right \} 
\to
\int_{ \frac{ i }{ J } }^{ \frac{ i+1 }{ J } }
		\varphi_{\frac{C}{H^{W_{s}}(M)}V^{ W_s }(M)} (x) dx  \quad ( n \to \infty ).
\end{align} 
With the limit $n \to \infty $, we have
\begin{align*}
P_{M^k} \left \{
\frac{ i }{ J } \le S(m^k) \le \frac{ i+1 }{ J }
\right \}
\Phi_{  V^*_{\pm} (W_{Y|X})  } 
 \left(
 R - \frac{ i+1 }{ J } 
 \right)
\to 
\int_{ \frac{ i }{ J } }^{ \frac{ i+1 }{ J } }
\varphi_{\frac{C}{H^{W_{s}}(M)}V^{ W_s }(M)} (x) dx 
\int_{ -\infty }^{ R - \frac{ i+1 }{ J } }
\varphi_{ V^*_{\pm} (W_{Y|X})} (y) dy .
\end{align*} 
So, taking the limit $n \to \infty$, we have 
\begin{align}
&
\sum_{ i= -LJ } ^{ i_0 } 
P_{M^k} \left \{
\frac{ i }{ J } \le S(m^k) \le \frac{ i+1 }{ J }
\right \}
\Phi_  { V^*_{-} (W_{Y|X}) }  
 \left(
 R - \frac{ i+1 }{ J } 
\right) \nonumber \\
&+
\sum_{ i= i_0 } ^{ LJ-1 } 
P_{M^k} \left \{
\frac{ i }{ J } \le S(m^k) \le \frac{ i+1 }{ J }
\right \}
\Phi_ { V^*_{+} (W_{Y|X}) }
\left( R - \frac{ i+1 }{ J } 
\right)	\\
\to&
\sum_{ i= -LJ } ^{ i_0 } \int_{ \frac{ i }{ J } }^{ \frac{ i+1 }{ J } }
\varphi_{\frac{C}{H^{W_{s}}(M)}V^{ W_s }(M)} (x) dx 
\int_{ -\infty }^{ R - \frac{ i+1 }{ J } }
\varphi_{ V^*_{\pm} (W_{Y|X})} (y) dy \nonumber \\
&+
\sum_{ i= i_0 } ^{ LJ-1 }  \int_{ \frac{ i }{ J } }^{ \frac{ i+1 }{ J } }
\varphi_{\frac{C}{H^{W_{s}}(M)}V^{ W_s }(M)} (x) dx 
\int_{ -\infty }^{ R - \frac{ i+1 }{ J } }
\varphi_{ V^*_{\pm} (W_{Y|X})} (y) dy. \Label{dis,so,st6,1}
\end{align}
When $ J \to \infty $, we can compute (\ref{dis,so,st6,1}) as:
\begin{align*}
&
\sum_{ i= -LJ } ^{ i_0 } \int_{ \frac{ i }{ J } }^{ \frac{ i+1 }{ J } }
\varphi_{\frac{C}{H^{W_{s}}(M)}V^{ W_s }(M)} (x) dx 
\int_{ -\infty }^{ R - \frac{ i+1 }{ J } }
\varphi_{ V^*_{-} (W_{Y|X})} (y) dy \nonumber \\
&
+
\sum_{ i= i_0 } ^{ LJ-1 }  \int_{ \frac{ i }{ J } }^{ \frac{ i+1 }{ J } }
\varphi_{\frac{C}{H^{W_{s}}(M)}V^{ W_s }(M)} (x) dx 
\int_{ -\infty }^{ R - \frac{ i+1 }{ J } }
\varphi_{ V^*_{+} (W_{Y|X})} (y) dy\\
\to&
 \int_{ -L }^{ R }
\left (
\int_{ -\infty }^{ R - x }
\varphi_{ V^*_{+} (W_{Y|X})} (y) dy
\right )
\varphi_{\frac{C}{H^{W_{s}}(M)}V^{ W_s }(M)} (x) dx \nonumber \\
& 
+
 \int_{ R }^{ L }
 \left (
 \int_{ -\infty }^{ R - x }
 \varphi_{ V^*_{-} (W_{Y|X})} (y) dy
 \right )
 \varphi_{\frac{C}{H^{W_{s}}(M)}V^{ W_s }(M)} (x) dx.
\end{align*}

Furthermore, when $ L \to \infty $, 
\begin{align*}
&
\int_{ -L }^{ R }
\left (
\int_{ -\infty }^{ R - x }
\varphi_{ V^*_{+} (W_{Y|X})} (y) dy
\right )
\varphi_{\frac{C}{H^{W_{s}}(M)}V^{ W_s }(M)} (x) dx 
 +
\int_{ R }^{ L }
\left (
\int_{ -\infty }^{ R - x }
\varphi_{ V^*_{-} (W_{Y|X})} (y) dy
\right )
\varphi_{\frac{C}{H^{W_{s}}(M)}V^{ W_s }(M)} (x) dx\\
\to&
\int_{ -\infty }^{ R }
\left (
\int_{ -\infty }^{ R - x }
\varphi_{ V^*_{+} (W_{Y|X})} (y) dy
\right )
\varphi_{\frac{C}{H^{W_{s}}(M)}V^{ W_s }(M)} (x) dx 
+
\int_{ R }^{ \infty }
\left (
\int_{ -\infty }^{ R - x }
\varphi_{ V^*_{-} (W_{Y|X})} (y) dy
\right )
\varphi_{\frac{C}{H^{W_{s}}(M)}V^{ W_s }(M)} (x) dx \\
=& \varepsilon(R).
\end{align*}
So, we obtain (\ref{dis,so,st6}).

\noindent{\it Step (vii):} Since $ \delta, \delta' > 0 $ are arbitrary, 
the combination of Steps (iv) and (v) yields
\begin{align}
\liminf_{ n \to \infty }\Pj[\phi|k, n]
\ge
\varepsilon(R). 
\end{align}

\section{The Comparison between Joint and Separation Scheme}\Label{S7}
\subsection{Formulation for separation coding}\Label{S71}
In this section, we compare the performance of the joint scheme with the performance of the separation scheme.
To discuss the separation scheme, we formulate a separation encoder and a separation decoder.
Firstly, we fix the input and output coding-lengths to be $k$ and $n$.
Then, we need to consider the encoded set $\{ 1 ,\cdots , A\}$ of source coding, which is also the message set of the channel coding.
Since the channel encoder does not know the source distribution,
it is natural to consider the average case with respect to the permutation on the set $\{ 1 ,\cdots , A \}$.
To handle such a permutation, we focus on the following triplet;
\begin{itemize}
\item A source encoder $\san{e}_{s,k,A} : {\cal M}^k \to \{ 1 ,\cdots , A\}$.\\
\item A source-channel mapping $ f_U:  \{ 1 ,\cdots , A \} \to \{ 1 ,\cdots , A \}$.\\
\item A channel encoder $ \san{e}_{c,A,n} :  \{ 1 ,\cdots , A \} \to {\cal X}^n $.
\end{itemize}
Then, our separation encoder is given as $\san{e}_{c,A,n} \circ f_U \circ \san{e}_{s,k,A}$.
The source-channel mapping $f_U$ is a random variable subject to the uniform distribution on the set of permutations on the set $\{ 1 ,\cdots , A \}$. 
To discuss the separation decoder, we consider
\begin{itemize}
\item A source decoder $\san{d}_{s,A,k} : \{ 1 ,\cdots , A \}\to {\cal M}^k$.\\
\item The inverse of the source-channel mapping $ f_U^{-1}:  \{ 1 ,\cdots , A \} \to \{ 1 ,\cdots , A \}$\\
\item A channel decoder $ \san{d}_{c,n,A} :  {\cal X}^n\to \{ 1 ,\cdots , A \} $.
\end{itemize}
So, our separation decoder is given as $\san{d}_{s,A,k}\circ f_U^{-1} \circ \san{d}_{s,A,k}$.
That is, our separation code is composed of 
$ (\san{e}^*_n, \san{d}^*_n):= (\san{e}_{c,A,n} \circ f_U \circ \san{e}_{s,k,A}  , 
\san{d}_{s,A,k}\circ f_U^{-1} \circ \san{d}_{s,A,k}) $.

Here, the source code $ (\san{e}_{s,k,A}, \san{d}_{s,A,k}) $ has the source coding rate
\begin{align}
R_{s}
:=
\frac{ \log A }{ k },
\end{align}
and the channel code $ (\san{e}_{c,A,n}, \san{d}_{c,n,A}) $ has the channel coding rate
\begin{align}
R_{c}
:=
\frac{ \log A }{ n }.
\end{align}
Then, the decoding error probability of the code $  (\san{e}^*_n, \san{d}^*_n) $
is given as the probability that the error occurs in the source coding or the channel coding.
Hence, the decoding error probability
$ \rom{P}_{\rm sep} (\san{e}^*_n, \san{d}^*_n) $ is defined as
\begin{align}
&\sum_{m \in { \cal M }^k:\san{d}_{s,A,k} \circ \san{e}_{s,k,A}( m ) \neq m} P_{ M^k }( m )
\nonumber \\
&+
\sum_{m \in { \cal M }^k:\san{d}_{s,A,k} \circ \san{e}_{s,k,A}( m ) = m}  P_{ M^k }( m )
W_{Y^n|X^n}( \{y: \san{d}_{c,n,A} (y) \neq f  \circ  \san{e}_{s,k,A}( m ) \} | \san{e}_{c,A,n}\circ f_U \circ  \san{e}_{s,k,A}( m ) ).
\end{align}
Since the source-channel mapping $f_U$ takes the value in the permutation on the set $\{ 1 ,\cdots , A \}$ subject to the uniform distribution,
it is natural to take the average with respect to the choice of $f_U$.
Hence, 
the value $\rom{P}_{\rm sep}[(\san{e}_{s,k,A}, \san{d}_{s,A,k}) ,(\san{e}_{c,A,n}, \san{d}_{c,n,A}) ] $
is defined as
the average of $ \rom{P}_{\rm sep} (\san{e}^*_n, \san{d}^*_n) $ with respect to this choice;
\begin{align}
\rom{P}_{\rm sep}[(\san{e}_{s,k,A}, \san{d}_{s,A,k}) ,(\san{e}_{c,A,n}, \san{d}_{c,n,A}) ] 
&:=
E_U \rom{P}_{\rm sep} (\san{e}^*_n, \san{d}^*_n). \Label{sep,ep} 
\end{align}

Let $\rom{P}_s(\san{e}_{s,k,A}, \san{d}_{s,A,k})$
be the decoding error probability of the source code $(\san{e}_{s,k,A}, \san{d}_{s,A,k})$,
and let
$\rom{P}_c(\san{e}_{c,A,n}, \san{d}_{c,n,A})$
be the decoding error probability of the channel code $(\san{e}_{c,A,n}, \san{d}_{c,n,A})$
with the message subject to the uniform distribution. 
Then, we have the following lemma.
\begin{lemma}\Label{L6}
The average 
$\rom{P}_{\rm sep}[(\san{e}_{s,k,A}, \san{d}_{s,A,k}) ,(\san{e}_{c,A,n}, \san{d}_{c,n,A}) ]
$ is calculated as
\begin{align}
\rom{P}_{\rm sep}[(\san{e}_{s,k,A}, \san{d}_{s,A,k}) ,(\san{e}_{c,A,n}, \san{d}_{c,n,A}) ]
=
\rom{P}_s(\san{e}_{s,k,A}, \san{d}_{s,A,k})
*
\rom{P}_c(\san{e}_{c,A,n}, \san{d}_{c,n,A}).\Label{Eq36}
\end{align}
\end{lemma}

\begin{proof}
From (\ref{sep,ep}), we have
\begin{align}
&\rom{P}_{\rm sep}[(\san{e}_{s,k,A}, \san{d}_{s,A,k}) ,(\san{e}_{c,A,n}, \san{d}_{c,n,A}) ] 
\nonumber \\
=&
E_U \rom{P}_{\rm sep} (\san{e}^*_n, \san{d}^*_n)\nonumber \\
=&
\sum_{m \in { \cal M }^k:\san{d}_{s,A,k} \circ \san{e}_{s,k,A}( m ) \neq m} P_{ M^k }( m )
\nonumber\\
&+
E_U \sum_{m \in { \cal M }^k:\san{d}_{s,A,k} \circ \san{e}_{s,k,A}( m ) = m}  P_{ M^k }( m )
W_{Y^n|X^n}( \{y: \san{d}_{c,n,A} (y) \neq f  \circ  \san{e}_{s,k,A}( m ) \} | \san{e}_{c,A,n}\circ f_U \circ  \san{e}_{s,k,A}( m ) ). \Label{com,pf} 
\end{align}
The second term of (\ref{com,pf}) can be calculated as follows. 
\begin{align}
&
E_U \sum_{m \in { \cal M }^k:\san{d}_{s,A,k} \circ \san{e}_{s,k,A}( m ) = m}  P_{ M^k }( m )
W_{Y^n|X^n}( \{y: \san{d}_{c,n,A} (y) \neq f  \circ  \san{e}_{s,k,A}( m ) \} | \san{e}_{c,A,n}\circ f_U \circ  \san{e}_{s,k,A}( m ) ) \nonumber\\
=&
\sum_{m \in { \cal M }^k:\san{d}_{s,A,k} \circ \san{e}_{s,k,A}( m ) = m}  P_{ M^k }( m )
\sum_{ a \in \{ 1 ,\cdots , A \}  }
\frac{ 1 }{ A ! }\sum_{U:f_U(m)=a}
W_{Y^n|X^n}( \{y: \san{d}_{c,n,A} (y) \neq a \} | \san{e}_{c,A,n}( a ) ) \nonumber\\
=&
\left (
1-\sum_{m \in { \cal M }^k:\san{d}_{s,A,k} \circ \san{e}_{s,k,A}( m ) \neq m} P_{ M^k }( m )
\right )
\frac{ 1 }{ A }\sum_{ a \in \{ 1 ,\cdots , A \}  }
W_{Y^n|X^n}( \{y: \san{d}_{c,n,A} (y) \neq a \} | \san{e}_{c,A,n}( a ) ).\Label{com,pf1}
\end{align}
Combining (\ref{com,pf}) and (\ref{com,pf1}), we have 
\begin{align*}
&\rom{P}_{\rm sep}[(\san{e}_{s,k,A}, \san{d}_{s,A,k}) ,(\san{e}_{c,A,n}, \san{d}_{c,n,A}) ] \\
=&
\sum_{m \in { \cal M }^k:\san{d}_{s,A,k} \circ \san{e}_{s,k,A}( m ) \neq m} P_{ M^k }( m )
\\
&+
\left (
1-\sum_{m \in { \cal M }^k:\san{d}_{s,A,k} \circ \san{e}_{s,k,A}( m ) \neq m} P_{ M^k }( m )
\right )
\frac{ 1 }{ A }\sum_{ a \in \{ 1 ,\cdots , A \}  }
W_{Y^n|X^n}( \{y: \san{d}_{c,n,A} (y) \neq a \} | \san{e}_{c,A,n}( a ) )\\
=&
\rom{P}_s(\san{e}_{s,k,A}, \san{d}_{s,A,k})
*
\rom{P}_c(\san{e}_{c,A,n}, \san{d}_{c,n,A}).
\end{align*}

\end{proof}

Under the fixed input and output coding-lengths $k$ and $n$, 
we minimize the above value 
$\rom{P}_{\rm sep}[(\san{e}_{s,k,A}, \san{d}_{s,A,k}) ,(\san{e}_{c,A,n}, \san{d}_{c,n,A}) ] $ as
\begin{align}
\rom{P}_{\rm sep }^{ * } (k,n:A) :=
\min_{(\san{e}_{s,k,A}, \san{d}_{s,A,k}) ,(\san{e}_{c,A,n}, \san{d}_{c,n,A})}\rom{P}_{\rm sep}[(\san{e}_{s,k,A}, \san{d}_{s,A,k}) ,(\san{e}_{c,A,n}, \san{d}_{c,n,A}) ] .
\end{align}
Here, since 
\begin{align}
&\min_{ (\san{e}_{s,k,A}, \san{d}_{s,A,k})}
\sum_{m \in { \cal M }^k:\san{d}_{s,A,k} \circ \san{e}_{s,k,A}( m ) \neq m} P_{ M^k }( m )
=\Ps(A ; P_{ M^k })\\
&\min_{(\san{e}_{s,k,A}, \san{d}_{c,n,A})}
\sum_{ a \in \{ 1 ,\cdots , A \}  }
W_{Y^n|X^n}( \{y: \san{d}_{c,n,A} (y) \neq a \} | \san{e}_{c,A,n}( a ) )
= \Pc(A ; W_{ Y^n| X^n }),
\end{align}
we have 
\begin{align}
&\rom{P}_{\rm sep }^{ * } (k,n:A) \nonumber \\
:=&
\min_{(\san{e}_{s,k,A}, \san{d}_{s,A,k}) ,(\san{e}_{c,A,n}, \san{d}_{c,n,A})}\rom{P}_{\rm sep}[(\san{e}_{s,k,A}, \san{d}_{s,A,k}) ,(\san{e}_{c,A,n}, \san{d}_{c,n,A}) ] \nonumber \\
=&
\Ps(A ; P_{ M^k }) * \Pc(A ; W_{ Y^n| X^n }) .\Label{sep,pro}
\end{align}
Note that for any two real numbers $ \alpha $ and $ \beta $, 
\begin{align}
\min_{\alpha,\beta}(\alpha * \beta)
=\min(\alpha)*\min(\beta). 
\end{align}
Considering the minimum with given value $A$, we have
\begin{align}
\rom{P}_{\rm sep }^{ * } (k,n)
= \min_{A} \rom{P}_{\rm sep }^{ * } (k,n;A).
\end{align}
Hereafter, we note the coding rate of the separation scheme $r_n$ as $ r_n := \frac{ k }{ n } $.
Additionally, we define
\begin{align}
\rom{K_{sep}}(n,\varepsilon)
:=
\inf \{ k | \rom{P}_{\rm sep }^{ * } (k,n) \le \varepsilon \}.
\end{align}

\begin{remark}
Many existing papers \cite{DAY2,ZA,VSM}
discussed the separation scheme, and they focused on the value 
$\rom{P}_s(\san{e}_{s,k,A}, \san{d}_{s,A,k}) * \rom{P}_c(\san{e}_{c,A,n}, \san{d}_{c,n,A}) $.
However, they did not give a rigorous derivation of this value.
The contribution of this subsection is derivation of this value from 
the formulation given here, which is rigorously shown as Lemma \ref{L6}.
\end{remark}

\subsection{Second order analysis}\Label{S72}
\subsubsection{Conditional additive channel case}
In this section, we evaluate the second order rate of the separation scheme. 
Using the $*$-product distribution $\tilde{\Phi} \left [
\frac{ C }{ H^{W_s}(M) } V^{ W_s } (M), 
V^{W_c} (X|Z)
\right ] $,
we have the following theorem for a conditional additive channel given by 
the transition matrix $ W_c $.
\begin{theorem}\Label{sep,so,th} 
The optimal transmission length
$\rom{K_{sep}}(n,\varepsilon)$ is asymptotically expanded as
\begin{align}
\lim_{n \to \infty}\frac{ \rom{K_{sep}}(n,\varepsilon)H^{ W_s }(M) - nC }{ \sqrt{n} }
=
\tilde{\Phi} \left [
\frac{ C }{ H^{W_s}(M) } V^{ W_s } (M), 
V^{W_c} (X|Z)
\right ] ^{-1}(\varepsilon). \Label{sep,so,th7}
\end{align} 
In other words, 
\begin{align}
\lim_{n \to \infty}\rom{P}_{\rm sep} 
\left (
n \frac{ C }{ H^{W_s}(M) } + \sqrt{n} \frac{ R }{ H^{W_s}(M) }, 
n
\right )
=
\varepsilon_{\rm sep}(R),
\end{align}
where
\begin{align}
\varepsilon_{\rm sep}(R)
:=
\tilde{\Phi} \left [
\frac{ C }{ H^{W_s}(M) } V^{ W_s } (M), 
V^{W_c} (X|Z)
\right ] (R). \Label{sep,so,th,err}
\end{align}
\end{theorem}

\begin{remark}
This theorem is an extension of the existing result \cite[Section V]{DAY}
to the case with Markovian source and a conditional additive channel.
\end{remark}

\begin{proof}
We assume that 
$\lim_{n \to \infty} \rom{P}_{\rm sep} (k,n) = \varepsilon$ 
and the intermediate set size of the separation code is $ A $. 
If $\Ps(A ; P_{ M^k })  \to \varepsilon_s$ and $ \Pc(A ; W_{ Y^n| X^n }) \to \varepsilon_c $ then $\varepsilon=\varepsilon_s * \varepsilon_c $.

The channel and source coding theorems for the Markovian case 
with the second order \cite[Theorems 10 and 21]{HW} guarantee the following relations
\begin{align}
\log A &= k H^{ W_s }(M) - \sqrt{ V^{ W_s }(M) }\sqrt{ k }\Phi^{-1}(\varepsilon_s)
+  o\left( \sqrt{k} \right), \Label{com,so,dis1}\\
\log A &=
n C +\sqrt{ V^{ W_c }(X|Z) }\sqrt{ n }\Phi^{-1}(\varepsilon_c)+  o\left( \sqrt{n} \right).
\end{align}
Hence, we have 
\begin{align}
k H^{ W_s }(M)
= n C + \sqrt{ V^{ W_s }(M) }\sqrt{ k }\Phi^{-1}(\varepsilon_s) + \sqrt{ V^{ W_c }(X|Z) }\sqrt{ n }\Phi^{-1}(\varepsilon_c) +  o\left( \sqrt{n} \right).
\end{align}

Since $ \frac{k}{n} = \frac{ C }{H^{ W_s }(M) } + o(n) $, 
\begin{align}
\frac{ k H^{ W_s }(M) }{ n }
=
C + 
\left (
\sqrt{ \frac{ C }{ H^{ W_s }(M) } V^{ W_s }(M) }\Phi^{-1}(\varepsilon_s)
+ \sqrt{ V^{ W_c }(X|Z) }\Phi^{-1}(\varepsilon_c) 
\right ) \sqrt{ \frac{ 1 }{ n } } +  o\left( \frac{ 1 }{ \sqrt{n} } \right). 
\end{align}
Optimizing the chose of $ A $, we have 
\begin{align}
\frac{ \rom{K_{sep}}(n,\varepsilon)H^{ W_s }(M) }{ n }
=
C + 
\max_{\varepsilon \ge \varepsilon_s * \varepsilon_c  }
\left (
\sqrt{ \frac{ C }{ H^{ W_s }(M) } V^{ W_s }(M) }\Phi^{-1}(\varepsilon_s)
 + \sqrt{ V^{ W_c }(X|Z) }\Phi^{-1}(\varepsilon_c) 
\right ) \sqrt{ \frac{ 1 }{ n } } 
+  o\left( \frac{ 1 }{ \sqrt{n} } \right). \Label{sep,so}
\end{align}
Hence, we have 
\begin{align}
\lim_{n \to \infty}\frac{ \rom{K_{sep}}(n,\varepsilon)H^{ W_s }(M) - nC }{ \sqrt{n} }
=
\max_{\varepsilon \ge \varepsilon_s * \varepsilon_c  }
\left (
\sqrt{ \frac{ C }{ H^{ W_s }(M) } V^{ W_s }(M) }\Phi^{-1}(\varepsilon_s)
 + \sqrt{ V^{ W_c }(X|Z) }\Phi^{-1}(\varepsilon_c) 
\right ).
\end{align}
\end{proof}

\subsubsection{Discrete memoryless channel case}
Using the $*$-product distribution, 
we evaluate the second order rate of separation coding in the discrete memoryless channel case.
\begin{theorem} \Label{sep,so,dis,th}
For the discrete memoryless channel give by a transition matrix W, we have  
\begin{align}
&\lim_{n \to \infty}
\frac{ \rom{K_{dm,sep}}(\varepsilon,n)H^{W_s}(M) - nC }{ \sqrt{n} }
=\varepsilon_{\rm sep}^{-1} (\varepsilon),
\end{align}
where
\begin{align}
\varepsilon_{\rm sep} (R)
	:=
\min
\left \{
\tilde{\Phi} \left [
\frac{ C }{ H^{W_s}(M) } V^{ W_s } (M), 
V^*_{+} (W_{Y|X})
\right ] (R), 
\tilde{\Phi} \left [
\frac{ C }{ H^{W_s}(M) } V^{ W_s } (M), 
V^*_{-} (W_{Y|X})
\right ] (R)
\right \}. 
\end{align}

\end{theorem}
\begin{remark}
The paper \cite[section V]{DAY} showed the same statement with the assumption 
$ V^*_+(W_{Y|X})= V^*_-(W_{Y|X}) $ and the source is independent and identical distribution.
Our contribution is removing the first assumption and generalizing it to Markovian source.
\end{remark}

\begin{proof}
We find that
\begin{align}
\varepsilon_{\rm sep}^{-1} (\varepsilon)
	=
\max
\left \{
\tilde{\Phi} \left [
\frac{ C }{ H^{W_s}(M) } V^{ W_s } (M), 
V^*_{+} (W_{Y|X})
\right ] ^{-1}(\varepsilon), 
\tilde{\Phi} \left [
\frac{ C }{ H^{W_s}(M) } V^{ W_s } (M), 
V^*_{-} (W_{Y|X})
\right ] ^{-1}(\varepsilon)
\right \}. 
\end{align}

	We assume that 
	$\lim_{n \to \infty} \rom{P}_{\rm sep} (k,n) = \varepsilon$ 
	and intermediate set size of separation code is $ A $. 
	If $\Ps(A ; P_{ M^k })  \to \varepsilon_s$ and $ \Pc(A ; W_{ Y^n| X^n }) \to \varepsilon_c $ then $\varepsilon=\varepsilon_s * \varepsilon_c $.
	The channel coding theorem with the second order \cite{7,HSO,12}
(Theorem \ref{dis,so,th} with uniform message of size $A$) guarantees that
	\begin{align}
\log A=	k H^{W_s}(M)
	+ \sqrt{n}\max
	\left \{
	\sqrt{V^*_+ (W_{Y|X})}\Phi^{-1}(\varepsilon_c),
	\sqrt{V^*_- (W_{Y|X})}\Phi^{-1}(\varepsilon_c)
	\right \}.\Label{com,so,dis2}
	\end{align}
Combining 
	(\ref{com,so,dis1}) and (\ref{com,so,dis2}), 
we obtain 
	\begin{align*}
	k H^{W_s}(M)-n C
	=&
	\sqrt{k} \sqrt{V^{ W_s } (M)} \Phi^{-1}(\varepsilon_s)
	+
	\sqrt{n} \max
	\left \{
	\sqrt{V^*_+ (W_{Y|X})}\Phi^{-1}(\varepsilon_c),
	\sqrt{V^*_- (W_{Y|X})}\Phi^{-1}(\varepsilon_c)
	\right \} \\
	=&
	\sqrt{n} 
\sqrt{ \frac{ C }{ H^{W_s}(M) } V^{ W_s } (M) }\Phi^{-1}(\varepsilon_s)
	+
	\sqrt{n} \max
	\left \{
	\sqrt{V^*_+ (W_{Y|X})}\Phi^{-1}(\varepsilon_c),
	\sqrt{V^*_- (W_{Y|X})}\Phi^{-1}(\varepsilon_c)
	\right \} 
	\end{align*}
because 
$	\frac{ k }{ n }
	=
	\frac{ C }{ H^{W_s}(M) } + o(n)$.
So, we have 
\begin{align}
&\lim_{n \to \infty}\frac{\rom{K_{dm,sep}} H^{W_s}(M)-nC}{\sqrt{n}} \nonumber \\
=& 
\max_{\varepsilon = \varepsilon_s * \varepsilon_c}
\left (
\sqrt{ \frac{ C }{ H^{W_s}(M) } V^{ W_s } (M) }\Phi^{-1}(\varepsilon_s)
+
\max
\left \{
\sqrt{V^*_+ (W_{Y|X})}\Phi^{-1}(\varepsilon_c),
\sqrt{V^*_- (W_{Y|X})}\Phi^{-1}(\varepsilon_c)
\right \}
\right) \nonumber \\
=& 
\max
\Bigg \{
\max_{\varepsilon = \varepsilon_s * \varepsilon_c}
\left (
\sqrt{ \frac{ C }{ H^{W_s}(M) } V^{ W_s } (M) }\Phi^{-1}(\varepsilon_s)
+
\sqrt{V^*_+ (W_{Y|X})}\Phi^{-1}(\varepsilon_c)
\right),\nonumber \\
&\max_{\varepsilon = \varepsilon_s * \varepsilon_c}
\left (
\sqrt{ \frac{ C }{ H^{W_s}(M) } V^{ W_s } (M) }\Phi^{-1}(\varepsilon_s)
+
\sqrt{V^*_- (W_{Y|X})}\Phi^{-1}(\varepsilon_c)
\right)\Bigg \}.
\end{align}
\end{proof}

\subsection{Comparison}\Label{sep,so,re}
Here, we compare the optimal error probability $\varepsilon(R) $
and the error probability $\varepsilon_{\rm sep}(R) $ of the separation scheme.
Since this comparison is based on the capacity
$C$, the source entropy rate $H^{W_s}(M) $,
the source variance $V^{ W_s } (M) $, and the channel variance,
the analysis of the conditional additive channel case
can be done as the same was as 
the analysis of the discrete memoryless channel case.
So, we discuss only the discrete memoryless channel case.

First, we compare the separation bound with the Kostina-Verd\'{u} bound $\varepsilon_{KV}(R)$ defined in 
\eqref{err,kv}, 
which is still not the tight bound in the joint source-channel scheme.
The property \eqref{sep,so,th1} implies the inequality
\begin{align}
\frac{\varepsilon_{\rm sep}(R) }{\varepsilon_{KV}(R) }\ge 1.
\end{align}
Here, the equality is attained if and only if 
$ V^{ W_s } (M) = 0 $, $H^{W_s}(M) = 0$, or $C = 0$.
When $H^{W_s}(M) = 0$, there is no information to be transmitted.
When $C = 0$, we cannot make any information transmission.
These two cases do not occur in a realistic case.
When $ V^{ W_s } (M) = 0 $, the distribution of the message source is uniform, which is not discussed 
in the joint source-channel coding.
So, we conclude that 
the separation scheme always has a larger decoding error probability than the joint source-channel scheme. 

As the opposite evaluation, we have the following lemma.
\begin{lemma}\Label{sep,so,lem2}
	We have
	\begin{align}
	\frac{\varepsilon_{\rm sep}(R) }{\varepsilon_{KV}(R)}\le
	\frac{2 \Phi(R_*) - \Phi(R_*)^2}
	{\Phi(R_*)}, \Label{EFRRR}
	\end{align}
	where
	\begin{align}
	R_*
	:=
\left\{
\begin{array}{ll}
\frac{ R }{ \sqrt{\frac{ C }{ H^{W_s}(M) } V^{ W_s } (M) + 
		V^*_{-} (W_{Y|X})} } & \hbox{ when } R \le 0 \\
\frac{ R }{ \sqrt{\frac{ C }{ H^{W_s}(M) } V^{ W_s } (M) + 
		V^*_{+} (W_{Y|X})} } & \hbox{ when } R > 0 .
\end{array}
\right.
	\end{align}
Under the conditions $ V^*_{+} (W_{Y|X})= V^*_{-} (W_{Y|X})$ and $R \le 0$,
	the equality holds if and only if
	$\frac{ C }{ H^{W_s}(M) } V^{ W_s } (M)  = V^*_{-} (W_{Y|X})$.
\end{lemma}

\begin{proof}
When $R \le 0 $, we have
	\begin{align}
{\varepsilon_{KV}(R)}
=&
\Phi_{\frac{ C }{ H^{W_s}(M) } V^{ W_s } (M) + 
		V^*_{-} (W_{Y|X})}(R) \\
\varepsilon_{\rm sep}(R)
\le &
\tilde{\Phi} \left [
\frac{ C }{ H^{W_s}(M) } V^{ W_s } (M), 
V^*_{+} (W_{Y|X})
\right ] (R).
	\end{align}
So, the inequality \eqref{sep,so,th1} of Lemma \ref{sep,so,lem} implies \eqref{EFRRR}.
We can show this inequality in the case of $R<0$.

When $ V^*_{+} (W_{Y|X})= V^*_{-} (W_{Y|X})$ and $R \le 0$,
Lemma \ref{sep,so,lem} guarantees that
	the equality holds if and only if
	$\frac{ C }{ H^{W_s}(M) } V^{ W_s } (M)  = V^*_{-} (W_{Y|X})$.
\end{proof}

When the variance of the information density is unique, i..e, 
$ V^*_{+} (W_{Y|X})= V^*_{-} (W_{Y|X})$,
Lemma \ref{sep,so,lem2} analytically determines the range of the ratio between 
the error probabilities with the joint and separation schemes.
For the general case,
combining Lemmas \ref{LH77} and \ref{sep,so,lem2}, we obtain the following lemma.
\begin{lemma}\Label{sep,so,lem5}
	We have
	\begin{align}
1\le	\frac{\varepsilon_{\rm sep}(R) }{\varepsilon(R)}
\le
\left\{
\begin{array}{ll}
 \frac{4 \Phi(R_*) - 2\Phi(R_*)^2}
	{\Phi(R_*)}
	& \hbox{ when } R < 0 \\
\frac{2 \Phi(R_*) - \Phi(R_*)^2}
	{\Phi(R^*) \Phi(R_*)}
	& \hbox{ when } R \ge 0,
\end{array}
\right.	\Label{EFRRX}
	\end{align}
	where 
$	R^*
	:=\frac{R}{\sqrt{\frac{ C }{ H^{W_s}(M) } V^{ W_s } (M)}}$.
\end{lemma}


\begin{remark}
The paper \cite[Section V]{DAY2} discussed a similar comparison as Lemma \ref{sep,so,lem2}
when the source is subject to an independent and identical distribution
and $V^*_{-} (W_{Y|X})=V^*_{+} (W_{Y|X})$.
Although the paper \cite[Section V]{DAY2} conjectured a similar statement as
Lemma \ref{sep,so,lem} via numerical calculation, they did not show it.
Hence, they could not analytically determine the range of the ratio between 
the error probabilities with the joint and separation schemes even when 
$ V^*_{+} (W_{Y|X})= V^*_{-} (W_{Y|X})$.
\end{remark}

\section{Discussion}
We have discussed the source-channel joint coding with the second order regime.
We have two open problems in this area.
One is the complete derivation of the second order coding rate in the general discrete memoryless case.
In this case, when the maximum and minimum variances has the same value,
the second order coding rate was derived by the paper \cite{DAY,DAY2}.
However, 
the general case had been remained as an open problem while a lower bound was obtained by Kostina and Verd\'{u} \cite{KV}.
Our optimal rate is strictly better than the lower bound by \cite{KV}.
To achieve such a better rate, we have invented a new random coding method, in which,
the distribution of the input alphabet is chosen according to 
the generation probability of the message.
Since the generation probability depends on the message in the joint coding regime, 
this improvement is very effective.
This coding method can be expected to another problem.
The second contribution is the derivation of the range of the ratio between 
the second order error probabilities of the joint and separation schemes.
The paper \cite{DAY2} derived an upper bound only by numerical calculation.
We have showed this conjecture analytically.
Further, we have given a rigorous formulation for the separation coding in Subsection \ref{S71}
while the error probability given in the RHS of \eqref{sep,ep} was used in many previous studies without rigorous derivation.

To obtain both main contributions, we have newly introduced two distribution families in Section \ref{S3}.
One is switched Gaussian convolution distributions and 
the other is $*$-product distribution.
Both distributions are defined as modifying the Gaussian distribution.
We have derived the notable relations among the cumulative distribution functions of these distributions and the Gaussian distribution.
The second contribution has been obtained from this kind of relations.  
Since these new distributions have operational meaning in this way,
we can expect that they will be applied to topics in information theory and related areas.

\section*{Acknowledgments}
MH is very grateful to Professor Vincent Y. F. Tan and 
Professor Shun Watanabe for helpful discussions and comments.
The works reported here were supported in part by 
JSPS Grants-in-Aid for Scientific Research (B) No. 16KT0017 and (A) No.17H01280,
the Okawa Research Grant
and Kayamori Foundation of Informational Science Advancement.

\appendices 

\section{Proof of Lemma \ref{sep,so,lem}}\Label{Ap3}
\noindent{\it Step (i):}
In this step, we prove the first inequality of (\ref{sep,so,th1}). 
Assume that $0< v_1,v_2 <\infty$.
Let $X$ and $Y$ be Gaussian random variables with mean $0$ and variance 
$v_1$ and $v_2$, respectively.
They are assumed to be independent of each other.
For a given real number $a$, 
we have
\begin{align}
\Phi_{v_1}
\left (a \right )
*
\Phi_{v_2}
\left ( R -a \right )
=
{\rm Pr}
\{
X \le a ~ {\rm or} ~ Y \le R-a 
\},
\end{align}
On the other hands, 
since $X+Y$ is a Gaussian random variable with variance $v_1+v_2$,
we have
\begin{align}
\Phi_{v_1+v_2}(R)
={\rm Pr}
\{
X + Y \le R 
\}.
\end{align}
Because 
$ \{
X + Y \le R 
\} 
\subsetneq \{
X \le a ~ {\rm or} ~ Y \le R-a
\}$, we have
\begin{align}
{\rm Pr}\{
X + Y \le R 
\} 
<
{\rm Pr} \{
X \le a ~ {\rm or} ~ Y \le R-a
\},
\end{align}
which implies that
\begin{align}
\Phi_{v_1+v_2}(R)
<
\Phi_{v_1}
\left (a \right )
*
\Phi_{v_2}
\left ( R -a \right ).
\end{align}
Taking the maximum with respect to $a$, we have
\begin{align}
\Phi_{v_1+v_2}(R)
<
\Psi[v_1,v_2](R).\Label{EH6-8}
\end{align}
Further, when $v_1$ or $v_2$ is zero, or $v_1$ or $v_2$ is infinity,
the equality holds in \eqref{EH6-8}.

\noindent{\it Step (ii):} In this step, 
we show the second inequality in (\ref{sep,so,th1}), and its equality condition.
For the proof, we define the new function $ \tilde{\varepsilon}(\varepsilon,y) $ for 
$ \varepsilon \in (0, 1) $ and $ y \le 0 $ as:
\begin{align}
\tilde{\varepsilon}(\varepsilon,y)
:=
\max_{ \varepsilon \ge \varepsilon_s * \varepsilon_c }
\frac{ 
	\Phi^{-1}(\varepsilon_s)
	+ \sqrt{y} \Phi^{-1}(\varepsilon_c) 
}{ \sqrt{1+ y }  }. 
\end{align}
Using this function, we can rewrite function $ \tilde{\Phi}[v_1, v_2]^{-1} (\varepsilon) $ 
as:
\begin{align}
\tilde{\Phi}[v_1, v_2]^{-1} (\varepsilon)
=
\sqrt{v_1 + v_2} 
\Phi^{-1}(\tilde{\varepsilon}(\varepsilon,y)),
\end{align}
where $ y = \frac{ v_2 }{ v_1 } $. 
Hence, the second inequality in (\ref{sep,so,th1}) and its equality condition
follow from the following two lemmas.

\begin{lemma}\Label{sep,so,th,l1}
For any $ y \ge 0 $, it holds that 
	\begin{align}
	\Phi^{ -1 }(\tilde{\varepsilon}(\varepsilon,y))
	\ge
	\sqrt{2} \Phi^{ -1 }(1 - \sqrt{1 - \varepsilon}).
	\end{align}
\end{lemma}

Hence, we obtain 
$R=\tilde{\Phi}[v_1, v_2]^{-1} (\varepsilon)
\ge 	\sqrt{2(v_1 + v_2)} \Phi^{ -1 }(1 - \sqrt{1 - \varepsilon})$.
So, we have
$ \Phi_{2(v_1 + v_2)}(R)\ge 1 - \sqrt{1 - \varepsilon}$,
which implies that
\begin{align}
\tilde{\Phi}[v_1, v_2](R)=\varepsilon \le
2 \Phi_{2(v_1 + v_2)}(R)
-\Phi_{2(v_1 + v_2)}(R)^2\Label{HE77}.
\end{align}
Due to the equality condition in Lemma \ref{sep,so,th,l1},
the equality holds in \eqref{HE77} only when $v_1=v_2 $. 
Conversely, we have the following lemma. 
\begin{lemma}\Label{sep,so,th,l2}
When $R \le 0$,
the equality
\begin{align}
\tilde{\Phi}[v_1,v_1](R)
=2 \Phi_{4 v_1}(R)
-\Phi_{4 v_1}(R)^2  \Label{HTT}
\end{align}
holds.
\end{lemma}

Hence, we can see that
the equality holds in \eqref{HE77} if and only if $v_1=v_2 $. 
Lemma \ref{sep,so,th,l1} is shown in Appendix \ref{Ap4},
and Lemma \ref{sep,so,th,l2} is shown in Appendix \ref{Ap5}.

\section{Proof of Lemma \ref{sep,so,th,l1} } \Label{Ap4}
It is sufficient to show the following two statements.
(1) For any $y>0$, we have 
\begin{align}
\Phi^{ -1 }(\tilde{\varepsilon}(\varepsilon,y))
=&
\max_{ \varepsilon = \varepsilon_s * \varepsilon_c }
\frac{ 
	\Phi^{-1}(\varepsilon_s)
	+ \sqrt{y} \Phi^{-1}(\varepsilon_c) 
}{ \sqrt{1+ y }  }\nonumber \\
\ge&
\max_{ \varepsilon = \varepsilon_s * \varepsilon_c } \inf_{ y\ge 0 }
\frac{ 
	\Phi^{-1}(\varepsilon_s)
	+ \sqrt{y} \Phi^{-1}(\varepsilon_c) 
}{ \sqrt{1+ y }  }  \Label{sep,so,th8}\\
=& 
\sqrt{2} \Phi^{ -1 }(1 - \sqrt{1 - \varepsilon}).
 \Label{sep,so,th8T}
\end{align}
(2) The maximum in \eqref{sep,so,th8} is realized only when $ \varepsilon_s=\varepsilon_c$.
Under this condition, the infimum 
$\inf_{ y\ge 0 }
\frac{ 
	\Phi^{-1}(\varepsilon_s)
	+ \sqrt{y} \Phi^{-1}(\varepsilon_c) 
}{ \sqrt{1+ y }  }$ is realized only when $y=1$.

The statement (1) implies \eqref{HE77},
and the statement (1) implies the necessarily condition for the equality in \eqref{HE77}.

\noindent{\it Step (i):} In this step, we will show the following relation for $\varepsilon \le \frac{3}{4} $.
\begin{align}
\max_{ \varepsilon = \varepsilon_s * \varepsilon_c } \inf_{ y\ge 0 }
\frac{ 
	\Phi^{-1}(\varepsilon_s)
	+ \sqrt{y} \Phi^{-1}(\varepsilon_c) 
}{ \sqrt{1+ y }  }
=
\max_{ \varepsilon = \varepsilon_s * \varepsilon_c , \varepsilon_s\le \frac{1}{2}, \varepsilon_c\le \frac{1}{2} }
- \sqrt{\Phi^{ -1 }( \varepsilon_s )^2 + \Phi^{ -1 }( \varepsilon_c )^2 }.
\Label{sep,so,th9}
\end{align}

Hence, it is sufficient to show that
\begin{align}
&
\max_{ \varepsilon = \varepsilon_s * \varepsilon_c } \inf_{ y\ge 0 }
\frac{ 1 }{ \sqrt{1 + y} }
\left (
\begin{array}{c}
1\\
\sqrt{y}
\end{array} 
\right )
\cdot
\left (
\begin{array}{c}
\Phi^{ -1 }( \varepsilon_s )\\
\Phi^{ -1 }( \varepsilon_c )
\end{array} 
\right ) \nonumber \\
=&
\max_{ \varepsilon = \varepsilon_s * \varepsilon_c }
- \sqrt{\Phi^{ -1 }( \varepsilon_s )^2 + \Phi^{ -1 }( \varepsilon_c )^2 }.\Label{sep,so,th5}
\end{align}

We rewrite the LHS of (\ref{sep,so,th5}) as
\begin{align}
\max_{ \varepsilon = \varepsilon_s * \varepsilon_c } \inf_{ y\ge 0 }
\frac{ 
	\Phi^{-1}(\varepsilon_s)
	+ \sqrt{y} \Phi^{-1}(\varepsilon_c) 
}{ \sqrt{1+ y }  }
=
\max_{ \varepsilon = \varepsilon_s * \varepsilon_c } \inf_{ y\ge 0 }
\frac{ 1 }{ \sqrt{1 + y} }
\left (
\begin{array}{c}
1\\
\sqrt{y}
\end{array} 
\right )
\cdot
\left (
\begin{array}{c}
\Phi^{ -1 }( \varepsilon_s )\\
\Phi^{ -1 }( \varepsilon_c )
\end{array} 
\right ), \Label{sep,so,th3}
\end{align}
where $\cdot$ is inner product of vector. 
The inside of the RHS of \eqref{sep,so,th3} is calculated as
\begin{align}
&\inf_{ y\ge 0 }
\frac{ 1 }{ \sqrt{1 + y} }
\left (
\begin{array}{c}
1\\
\sqrt{y}
\end{array} 
\right )
\cdot
\left (
\begin{array}{c}
\Phi^{ -1 }( \varepsilon_s )\\
\Phi^{ -1 }( \varepsilon_c )
\end{array} 
\right )\nonumber \\
=&
\left \{
\begin{array}{ ll }
- \sqrt{\Phi^{ -1 }( \varepsilon_s )^2 + \Phi^{ -1 }( \varepsilon_c )^2 } 
& {\rm when} \quad \Phi^{ -1 }( \varepsilon_s ) \le 0 \  {\rm and} \  \Phi^{ -1 }( \varepsilon_c ) \le 0 \\
\Phi^{ -1 }( \varepsilon_s )
&{\rm when} \quad \Phi^{ -1 }( \varepsilon_s ) \le 0 \ {\rm and} \  \Phi^{ -1 }( \varepsilon_c ) \ge 0 \\
\Phi^{ -1 }( \varepsilon_c )
&{\rm when} \quad \Phi^{ -1 }( \varepsilon_s ) \ge 0 \  {\rm and} \  \Phi^{ -1 }( \varepsilon_c ) \le 0 .
\end{array}
\right . \Label{sep,so,th4}
\end{align}
Since  $\varepsilon \le \frac{ 3 }{ 4 }$, 
either 
$ \Phi(\varepsilon_s) $ or $ \Phi(\varepsilon_c)$ 
is negative. 
Hereafter, we will consider the maximum value 
of (\ref{sep,so,th4}) under the condition 
$\varepsilon = \varepsilon_s * \varepsilon_c$.

When $\varepsilon \le\frac{1}{2}$, 
we have $\varepsilon_s ,\varepsilon_c \le \frac{1}{2}$, which implies \eqref{sep,so,th5}.
So, we consider the case when $\frac{1}{2}<\varepsilon\le \frac{3}{4}$, which has the above three cases.
First, we consider the case when 
$\Phi^{ -1 }(\varepsilon_s ) \le 0$ 
and $\Phi^{ -1 }( \varepsilon_c ) \ge 0 $.
Then, we have
\begin{align}
\max_{\substack{\varepsilon_s,\varepsilon_c : 
		 \varepsilon = \varepsilon_s * \varepsilon_c,\\
		\Phi^{ -1 }(\varepsilon_s ) \le 0,\\
		\Phi^{ -1 }( \varepsilon_c ) \ge 0}}
\Phi^{ -1 }( \varepsilon_s )
=&
\max_{\substack{\varepsilon_s,\varepsilon_c : 
		 \varepsilon = \varepsilon_s * \varepsilon_c, \nonumber \\
		\varepsilon_s \le \frac{ 1 }{ 2 }, \varepsilon_c \ge \frac{ 1 }{ 2 }}}
\Phi^{ -1 }( \varepsilon_s ) \nonumber \\
=&
\max_{\varepsilon_s,\varepsilon_c : 
	0 \le \varepsilon_s \le 2\varepsilon - 1}
\Phi^{ -1 }( \varepsilon_s ) \nonumber \\
=&
\Phi^{ -1 }( 2 \varepsilon - 1 ).
\end{align}
That is, the maximum value is attained 
when $\Phi^{ -1 }( \varepsilon_s ) = \Phi^{-1}(2 \varepsilon - 1)$ and 
$\Phi^{-1}( \varepsilon_c )=0 $. 

We obtain the same equation in the case when 
$\Phi^{ -1 }(\varepsilon_s ) \ge 0$ and $\Phi^{ -1 }( \varepsilon_c ) \le 0 $.
Hence,
we find that that the maximum of the RHS of \eqref{sep,so,th4} equals
$
\max_{ \varepsilon = \varepsilon_s * \varepsilon_c , \varepsilon_s\le \frac{1}{2}, \varepsilon_c\le \frac{1}{2} }
- \sqrt{\Phi^{ -1 }( \varepsilon_s )^2 + \Phi^{ -1 }( \varepsilon_c )^2 }$,
which implies \eqref{sep,so,th9}.

\noindent{\it Step (ii):} In this step, 
when $\varepsilon \le \frac{3}{4}$,
we will show the following equation.
Also we will show that the following maximum is realized if and only if $\varepsilon_s = \varepsilon_c $.
Since the discussion of Step (i) shows that
Under this condition, 
the infimum 
$\inf_{ y\ge 0 }
\frac{ 
	\Phi^{-1}(\varepsilon_s)
	+ \sqrt{y} \Phi^{-1}(\varepsilon_c) 
}{ \sqrt{1+ y }  }$ is realized only when $y=1$.
These discussions show the desired statements (1) and (2) with $\varepsilon \le \frac{3}{4}$.
\begin{align}
\max_{ \varepsilon = \varepsilon_s * \varepsilon_c , \varepsilon_s\le \frac{1}{2}, \varepsilon_c\le \frac{1}{2} }
- \sqrt{\Phi^{ -1 }( \varepsilon_s )^2 + \Phi^{ -1 }( \varepsilon_c )^2 }
=
\sqrt{2} \Phi^{ -1 }(1 - \sqrt{1 - \varepsilon}).\Label{eq30}
\end{align}

For notation, we define 
the function for $\varepsilon_s $ and $\varepsilon_c$ as: 
\begin{align}
A(\varepsilon_s, \varepsilon_c)
:=
\sqrt{\Phi^{ -1 }( \varepsilon_s )^2 + \Phi^{ -1 }( \varepsilon_c )^2 }.
\end{align}
Then, we can find that 
\begin{align}
\lim_{\varepsilon_s \to 0}A(\varepsilon_s, \varepsilon_c)
=
\lim_{\varepsilon_c \to 0}A(\varepsilon_s, \varepsilon_c)
=
\infty,
\end{align}
and 
$\max_{ \varepsilon = \varepsilon_s * \varepsilon_c , \varepsilon_s\le \frac{1}{2}, \varepsilon_c\le \frac{1}{2} }
  A(\varepsilon_s, \varepsilon_c)$
is monotonically decreasing function 
of $\varepsilon $. 
Hence, the relation \eqref{eq30} is equivalent to 
\begin{align}
\min_{ \varepsilon_s , \varepsilon_c:  
A(\varepsilon_s, \varepsilon_c)
= -\sqrt{2} \Phi^{ -1 }(1 - \sqrt{1 - \varepsilon})}
 \varepsilon_s * \varepsilon_c
= \varepsilon. \Label{eq32}
\end{align}

Choosing $a:=-\sqrt{2} \Phi^{ -1 }(1 - \sqrt{1 - \varepsilon})>0$,
we write 
\begin{align*}
\Phi^{ -1 }(\varepsilon_s) &= a \cos \theta,\\
\Phi^{ -1 }(\varepsilon_c) &= a \sin \theta,
\end{align*}
for certain 
$\pi \le \theta \le \frac{ 3 }{ 2 } \pi$.
Now, to regard $\varepsilon_s * \varepsilon_c$ as a function of $\theta$, 
we define 
\begin{align}
f (\theta)
:=
\Phi(a \cos \theta) * \Phi(a \sin \theta), 
\end{align}
and hereafter 
we will find $ \theta $ which minimize $ f (\theta) $. 
Calculating the derivative, 
we have 
\begin{align*}
&\frac{ d f (\theta)}{ d \theta } \nonumber \\
=&
\frac{ d }{ d \theta } 
[\Phi(a \cos \theta) + \Phi(a \sin \theta) -  \Phi(a \cos \theta) \Phi(a \sin \theta)]\\
=&
-a \sin \theta \Phi'(a \cos \theta) + a \cos \theta \Phi'(a \sin \theta)
+a \sin \theta \Phi'(a \cos \theta) \Phi(a \sin \theta)
-  a \cos \theta \Phi(a \cos \theta) \Phi'(a \sin \theta)\\
=&
- a \sin \theta \Phi'(a \cos \theta) (1- \Phi(a \sin \theta))
+ a \cos \theta \Phi'(a \sin \theta) (1 - \Phi(a \cos \theta)).
\end{align*}
Now, we define
\begin{align*}
g_a(x)
:=
-a \sqrt{1 - x^2}
\Phi'(a x)  (1- \Phi(- a \sqrt{1 - x^2})). 
\end{align*}
Because 
$ \Phi'(x) = \frac{ 1 }{ \sqrt{2 \pi} }e^{- \frac{ x^2 }{ 2 } }  $ 
and $ \Phi(x) $ 
is a monotonically increasing function 
for $ x<0 $, 
we find that 
$ g_a(x) $ is a monotonically increasing function 
for $ x<0 $.
Since 
\begin{align}
\frac{ d f(\theta)}{ d \theta } 
=
-g_a(\cos \theta) + g_a(\sin \theta),
\end{align}
the derivative test chart of 
$ f (\theta) $ is given as 
follows. 
\begin{align}
\begin{array}{|c| * 5 {c|}} \hline
\theta & \pi & \cdots & \frac{ 5 }{ 4 } \pi & \cdots & \frac{ 3 }{ 2 }\pi \\ \hline
\frac{ d  f (\theta)}{ d \theta } & & - & & + & \\ \hline
f (\theta) & &\searrow & & \nearrow &  \\ \hline
\end{array}
\end{align}
Hence, 
when $ \theta = \frac{ 5 }{ 4 }\pi $ i.e., $\varepsilon_s=\varepsilon_c  $, 
$ f (\theta) $ is minimized.
Therefore,
when $ (\varepsilon_s, \varepsilon_c ) $ satisfies 
$\varepsilon = \varepsilon_s*\varepsilon_c$ and $\varepsilon_s=\varepsilon_c  $,
the minimum \eqref{eq32} is attained.
So, we have 
$\varepsilon_s=\varepsilon_c = 1 - \sqrt{1 - \varepsilon}$, 
which means \eqref{eq32}.


\noindent{\it Step (iii):} 
In this step, when $\varepsilon > \frac{3}{4}$,
we will show the following equation.
Also we will show that the following maximum is realized if and only if $\varepsilon_s = \varepsilon_c $.
Since the discussion of Step (i) shows that
Under this condition, 
the infimum 
$\inf_{ y\ge 0 }
\frac{ 
	\Phi^{-1}(\varepsilon_s)
	+ \sqrt{y} \Phi^{-1}(\varepsilon_c) 
}{ \sqrt{1+ y }  }$ is realized only when $y=1$.
These discussions show the desired statements (1) and (2) with $\varepsilon > \frac{3}{4}$.
\begin{align}
\max_{ \varepsilon \ge \varepsilon_s * \varepsilon_c } 
\inf_{ y\ge 0 }
\frac{ 
	\Phi^{-1}(\varepsilon_s)
	+ \sqrt{y} \Phi^{-1}(\varepsilon_c) 
}{ \sqrt{1+ y }  }
=
\sqrt{2} \Phi^{ -1 }(1 - \sqrt{1 - \varepsilon}).\Label{eq3T}
\end{align}

Since $ \varepsilon > \frac{3}{4}$,
we have four cases.
(1) $\Phi^{ -1 }( \varepsilon_s ) \le 0 $ and $ \Phi^{ -1 }( \varepsilon_c ) \le 0 $.
(2) $\Phi^{ -1 }( \varepsilon_s ) > 0 $ and $ \Phi^{ -1 }( \varepsilon_c ) \le 0 $.
(3) $\Phi^{ -1 }( \varepsilon_s ) \le 0 $ and $ \Phi^{ -1 }( \varepsilon_c ) > 0 $.
(4) $\Phi^{ -1 }( \varepsilon_s ) > 0 $ and $ \Phi^{ -1 }( \varepsilon_c ) > 0 $.
The infinum 
$\inf_{ y\ge 0 }
\frac{ 
	\Phi^{-1}(\varepsilon_s)
	+ \sqrt{y} \Phi^{-1}(\varepsilon_c) 
}{ \sqrt{1+ y }  }
$ is negative except for the case (4).
So, the maximum with respect to 
$\varepsilon_s $ and $ \varepsilon_c$ under the condition
$\varepsilon \ge \varepsilon_s * \varepsilon_c$
is realized in the case (4).
In the case (4),
we have
\begin{align}
\inf_{ y\ge 0 }
\frac{ 
	\Phi^{-1}(\varepsilon_s)
	+ \sqrt{y} \Phi^{-1}(\varepsilon_c) 
}{ \sqrt{1+ y }  }
=\min(\Phi^{-1}(\varepsilon_s),\Phi^{-1}(\varepsilon_c)).
\Label{EH688}
\end{align}
The maximum of the RHS of \eqref{EH688} with the condition
$\varepsilon \ge \varepsilon_s * \varepsilon_c$
is realized when 
$\varepsilon = \varepsilon_s * \varepsilon_c$ and $\varepsilon_s = \varepsilon_c$.
Solving the equation $\varepsilon = \varepsilon_s * \varepsilon_s$, we have
\begin{align}
\max_{ \varepsilon \ge \varepsilon_s * \varepsilon_c } 
\min(\Phi^{-1}(\varepsilon_s),\Phi^{-1}(\varepsilon_c))
=
\sqrt{2} \Phi^{ -1 }(1 - \sqrt{1 - \varepsilon}).
\Label{EH689}
\end{align}
So, the combination of \eqref{EH688} and \eqref{EH689} yields  \eqref{eq3T}.
 
\section{Proof of Lemma \ref{sep,so,th,l2}}\Label{Ap5}
It is sufficient to show the case with $v_1=1$.
We set the function
	\begin{align}
	f(s)	:= 	\Phi(s) * \Phi(R-s).
	\end{align}
	That is, it is sufficient to show that
	the minimum $ \min_{s} f(s)$ is realized when $ s=\frac{R}{2}$
	because $\Phi(\frac{R}{2}) * \Phi(\frac{R}{2})$ equals the RHS of \eqref{HTT}. 
	
Calculating the derivative, we have
\begin{align*}
	\frac{ d f(s) }{ ds }
	=&
	\frac{ d }{ ds }[\Phi(s) + \Phi(R-s) - \Phi(s) \Phi(R-s)]\\
	=&
	\Phi'(s) - \Phi'(R-s) - \Phi'(s)\Phi(R-s) + \Phi(s) \Phi'(R-s)\\
	=&
	\Phi'(s)(1- \Phi(R-s)) 
	- \Phi'(R-s)(1-\Phi(s)) \\
	=&
	\Phi'(s)\Phi(s-R) 	- \Phi'(R-s)\Phi(-s) .
\end{align*}
The function $x \mapsto \Phi'(x)\Phi(x-R)$ is a monotonically increasing function 
for $ x<0 $. 
So, we find that
$	\frac{ d f(s) }{ ds } \le 0$ for $ s \in [R,\frac{R}{2}]$
and
$	\frac{ d f(s) }{ ds } \ge 0$ for $ s \in [\frac{R}{2},0]$.
Further, when $ s <R$,
$ |s| > |R-s|$, which implies that $ \Phi'(s)< \Phi'(R-s)$.
In this case, we have $s -R < -s$, which implies $ \Phi(s-R)< \Phi(-s) $.
So, we obtain $\Phi'(s)\Phi(s-R) - \Phi'(R-s)\Phi(-s) <0$, i.e., 
$	\frac{ d f(s) }{ ds } < 0$.
Similarly, when $ s >0$, we can show the inequality
$	\frac{ d f(s) }{ ds } > 0$.
Therefore, the minimum $ \min_{s} f(s)$ is realized when $ s=\frac{R}{2}$.

\section{Proof of Lemma \ref{L-11}}\Label{A10}
First, we set
\begin{align}
{\cal L} := \{(m, x, y) \in ({\cal M, X, Y}) | P_M(m)W_{Y|X=x} (y)
\le c Q_Y (y)
 \}.
\end{align}
For each $(m, x) \in ({\cal M, X})$, we define
\begin{align}
{\cal B}(m, x) := \{ y \in {\cal Y} | (m, x, y) \in {\cal L} \}. 
\end{align}
Also, for decoder $\varphi$ and each $m \in {\cal M}$, we define
\begin{align}
{\cal D}(m) := \{y \in {\cal Y} | \varphi(y) = m \}. 
\end{align}
In addition, we define $P_{X|M}$ so that
\begin{align}
P_{X|M}(x|m) =
	\begin{cases}
	0 & x \neq \san{e}(m)\\
	1 & x= \san{e}(m).
	\end{cases}
\end{align}
Using this notation, we define
\begin{align}
P_{MX}(m, x) &:= P_M(m)P_{X|M}(x|m) , \\
P_{MXY}(m, x, y)&:= P_M(m)P_{X|M}(x|m)W_{Y|X=\san{e}(m)} (y). 
\end{align}

Then, 
\begin{align}
&\sum_m P_M(m)W_{Y|X=\san{e}(m)} 
\{ P_M(m)W_{Y|X=\san{e}(m)} (Y)
\le c Q_Y (Y)
\} \nonumber\\
=& \sum_{(m, x, y) \in {\cal L}} P_{MXY} (m, x, y) \nonumber\\
=& \sum_{(m, x) \in {\cal M, X}} \sum_{y \in {\cal B}(m, x)}P_{MX}(m, x) W_{Y|X}(y|x) 
\nonumber\\
=& \sum_{(m, x) \in {\cal M, X}} \sum_{y \in {\cal B}(m, x) \cap {\cal D}(m)}
		P_{MX}(m, x) W_{Y|X}(y|x)
	+\sum_{(m, x) \in {\cal M, X}} \sum_{y \in {\cal B}(m, x)\cap {\cal D}^c(m)}
			P_{MX}(m, x) W_{Y|X}(y|x) \nonumber\\
\le & \sum_{(m, x) \in {\cal M, X}} \sum_{y \in {\cal B}(m, x) \cap {\cal D}(m)}
	P_{MX}(m, x) W_{Y|X}(y|x)
	+\sum_{(m, x) \in {\cal M, X}} \sum_{y \in {\cal D}^c(m)}
	P_{MX}(m, x) W_{Y|X}(y|x) \nonumber\\
=& \sum_{(m, x) \in {\cal M, X}} \sum_{y \in {\cal B}(m, x) \cap {\cal D}(m)}P_{MX}(m, x) W_{Y|X}(y|x) + \Pjs[\phi|P_M, W_{Y|X}]. \Label{sbcl1p1}
\end{align}
The last equality follows since the error probability can be written as
\begin{align*}
\Pjs[\phi|P_M, W_{Y|X}] = \sum_{(m, x) \in {\cal M, X}} \sum_{y \in {\cal D}^c(m)}P_{MX}(m, x) W_{Y|X}(y|x). 
\end{align*}
We notice here that
\begin{align*}
 P_M(m)W_{Y|X=\san{e}(m)} (Y)
\le c Q_Y (Y)
\end{align*}
for $y \in {\cal B}(m, x)$. By substituting this into (\ref{sbcl1p1}), the first term of (\ref{sbcl1p1}) is
\begin{align*}
&\sum_{(m, x) \in {\cal M, X}} \sum_{y \in {\cal B}(m, x) \cap {\cal D}(m)}c P_{X|M}(x|m)Q_Y (y)\\
\le & \sum_{(m, x) \in {\cal M, X}} \sum_{y \in {\cal D}(m)}c P_{X|M}(x|m)Q_Y (y)\\
=& c \sum_{m \in {\cal M}} \sum_{y \in {\cal D}(m)}Q_Y (y)\\
=& c \sum_{m \in {\cal M}}Q_Y ({\cal D}(m)) = c,
\end{align*}
which implies (\ref{2}).

\bibliographystyle{IEEE}

\end{document}